\documentclass[a4paper,USenglish,cleveref,autoref]{article} 
\usepackage[utf8]{inputenc}
\usepackage{amsmath,braket,mathrsfs,url,epsfig,amsfonts}
\usepackage{amssymb,fancyhdr}
\usepackage{latexsym}
\usepackage{amsthm}
\usepackage{hyperref}
\usepackage{cleveref}
\usepackage{babel}
\usepackage{babelbib}
\usepackage[all]{xy}
\usepackage[dvipsnames]{xcolor}
\usepackage{algorithm}
\usepackage{algorithmicx}
\usepackage{algpseudocode}
\usepackage{epstopdf}
\usepackage{lineno}
\usepackage{bbm}
\usepackage{mathtools}
\usepackage{sectsty} 
\usepackage{subcaption}
\usepackage{tikz}
\usepackage{tikz-cd}
\usetikzlibrary{cd}
\usetikzlibrary{patterns}
\usetikzlibrary{positioning}
\usetikzlibrary{arrows.meta}
\usepackage{booktabs}
\usepackage{todonotes}
\usepackage{mathtools}
\DeclarePairedDelimiter\ceil{\lceil}{\rceil}
\DeclarePairedDelimiter\floor{\lfloor}{\rfloor}

\definecolor{niceyellow}{rgb}{0.99,0.78,0.07}

\theoremstyle{definition}
\newtheorem{definition}{Definition}
\theoremstyle{definition}

\newtheorem{remark}[definition]{Remark}
\theoremstyle{plain}
\newtheorem{theorem}[definition]{Theorem}
\newtheorem{lemma}[definition]{Lemma}

\newtheorem{proposition}[definition]{Proposition}

\newcommand{\R}{\mathbb{R}}
\newcommand{\N}{\mathbb{N}}
\newcommand{\kmax}{K}

\newcommand{\firstcovername}{\mathfrak{X}}

\newcommand{\nerve}{\textnormal{Nrv}}
\newcommand{\covonly}{\textnormal{Cov}}
\newcommand{\cov}[2]{\covonly_{#1,#2}}
\newcommand{\covreg}[2]{\covonly_{#1}(#2)}
\newcommand{\bigcovonly}{\mathfrak{Cov}}
\newcommand{\bigcov}[2]{\bigcovonly_{#1,#2}}
\newcommand{\tcovonly}{\textnormal{S-}{\delonly}}
\newcommand{\tcov}[2]{\tcovonly_{#1,#2}}
\newcommand{\voronly}{\textnormal{Vor}}
\newcommand{\vor}[1]{\voronly_{#1}}
\newcommand{\vorcovonly}{\mathfrak{Vor}}
\newcommand{\vorcov}[2]{\vorcovonly_{#1,#2}}
\newcommand{\vorreg}[1]{\textnormal{Vor}(#1)}
\newcommand{\delonly}{\textnormal{Del}}

\newcommand{\del}[2]{\delonly_{#1,#2}}
\newcommand{\tdelonly}{\widetilde{\delonly}}
\newcommand{\tdel}[2]{\tdelonly_{#1,#2}}

\newcommand{\deltaspaceonly}{\Delta}
\newcommand{\deltaspace}[2]{\deltaspaceonly{#1}
}
\newcommand{\rhomonly}{\textnormal{{Rhomb}}}

\newcommand{\trhomonly}{\textnormal{S-}{\rhomonly}}
\newcommand{\trhom}[2]{\trhomonly_{#1,#2}}

\newcommand{\mm}[1]        {\ifmmode{#1}\else{\mbox{\(#1\)}}\fi}

\def \Xp{{A}}                     

\newcommand{\qv}           {\mm{\tilde{A}}} 
\newcommand{\Xin}          {\mm{{\Xp}_{\it in}}}
\newcommand{\Xon}          {\mm{{\Xp}_{\it on}}}
\newcommand{\Xout}         {\mm{{\Xp}_{\it out}}}
\newcommand{\Arr}[1]       {\mm{{Arr}{}}}

\newcommand{\Rhomboid}[2]  {\mm{{\rhomonly}_{#1,#2}}}

\newcommand{\vx}[1]        {\mm{{\rho}{_{#1}}}}  
\newcommand{\xin}[1]       {\mm{{\Xp}_{\it in}{({#1})}}}
\newcommand{\xon}[1]       {\mm{{\Xp}_{\it on}{({#1})}}}

\newcommand{\altkmin}{k_{\mathrm{min}}}
\newcommand{\altkmax}{k_{\mathrm{max}}}

\newcommand{\ignore}[1]{}

\title{Computing the multicover bifiltration} 

\author{Ren\'{e} Corbet\thanks {KTH Royal Institute of Technology, corbet@kth.se},
Michael Kerber\thanks {Graz University of Technology, kerber@tugraz.at},
Michael Lesnick\thanks{University at Albany, SUNY, mlesnick@albany.edu},
Georg Osang\thanks{Institute of Science and Technology Austria, georg.osang@ist.ac.at}}

\begin{document}

\maketitle

\begin{abstract}
Given a finite set $A\subset\mathbb{R}^d$, let $\cov{r}{k}$ denote the set of all points within distance $r$ to at least $k$ points of $A$.  Allowing $r$ and $k$ to vary, we obtain a 2-parameter family of spaces that grow larger when $r$ increases or $k$ decreases, called the \emph{multicover bifiltration}.  Motivated by the problem of computing the homology of this  bifiltration, we introduce two closely related combinatorial bifiltrations, one polyhedral and the other simplicial, which are both topologically equivalent to the multicover bifiltration and far smaller than a \v Cech-based model considered in prior work of Sheehy.  Our polyhedral construction is a bifiltration of the \emph{rhomboid tiling} of Edelsbrunner and Osang, and can be efficiently computed using a variant of an algorithm given by these authors.  Using an implementation for dimension 2 and 3, we provide experimental results.  Our simplicial construction is useful for understanding the polyhedral construction and proving its correctness. 
\end{abstract}

\section{Introduction}\label{sec:introduction}

Let $A$ be a finite subset of $\R^d$, whose points we call \emph{sites}. For $r\in [0,\infty)$ and an integer $k\in \mathbb N=\{0,1,2,\ldots\}$, we define
\[
\cov{r}{k}:=\left\{  b\in\mathbb{R}^d \mid \textnormal{ } ||b-a||\leq r \textnormal{ for at least $k$ sites } a\in A  \right\}.
\]
Thus, $\cov{r}{k}$ is the union of all $k$-wise intersections of closed balls
of radius $r$ centered at the sites; see Figure~\ref{fig:kfold}.   
Define a \emph{bifiltration} to be   
 a collection of sets \[C:=(C_{r,k})_{(r,k)\in [0,\infty)\times\N}\] such that
$C_{r,k}\subseteq C_{r',k'}$ whenever $r\leq r'$ and $k\geq k'$.
Clearly, the sets \[\covonly:=(\cov{r}{k})_{(r,k)\in [0,\infty)\times\N}\] form a bifiltration.  This is known as the \emph{multicover bifiltration}.  It arises natural\-ly in topological data analysis (TDA), and specifically, in the topological analysis of data with outliers or non-uniform density \cite{Chazal2011,Edelsbrunner2018,Sheehy2012}.

We wish to study the topological structure of the bifiltration $\covonly$ algorithmi\-cally in practical applications, via 2-parameter persistent homology \cite{carlsson2009theory}.  For this, the natural first step is to compute a \emph{combinatorial model} of $\covonly$, that is, a purely combinatorial bifiltration $C$ which is topologically equivalent to $\covonly$.  This step is the focus of the present paper.  For computational efficiency, $C$ should not be too large. 

 In fact, we propose two closely related combinatorial models $C$, one polyhe\-dral and one simplicial.  The polyhedral model is a bifiltration of the \emph{rhomboid tiling}, a polyhedral cell complex in $\R^{d+1}$ recently introduced by Edelsbrunner and Osang to study the multicover bifiltration \cite{Edelsbrunner2018}.  Edelsbrunner and Osang have given an efficient algorithm for computing the rhomboid tiling~\cite{EdOs20}, and this adapts readily to compute our bifiltration.  We use the simplicial model to prove that the polyhedral model is topologically equivalent to $\covonly$.

\begin{figure}
\centering
\begin{tikzpicture}[scale=1.4]
	\tikzstyle{point}=[circle,thick,draw=black,fill=black,inner sep=0pt,minimum width=4pt,minimum height=4pt]
	\draw (-2.2, 2) rectangle (2, -1.6);
	\begin{scope}
		\clip ( 0.5, .3) circle (.8);
		\clip ( -1, 0) circle (.8);
		\fill[color=niceyellow] (-2,2)rectangle (2,-2);
	\end{scope}
	\begin{scope}
		\clip (-0.5, 0) circle (.8);
		\clip ( -1, 0) circle (.8);
		\fill[color=niceyellow] (-2,2)rectangle (2,-2);
	\end{scope}
	\begin{scope}
		\clip (-0.5, 0) circle (.8);
		\clip ( 0.5, .3) circle (.8);
		\fill[color=niceyellow] (-2,2)rectangle (2,-2);
	\end{scope}
	\begin{scope}
		\clip ( -1.2, .6) circle (.8);
		\clip ( -1, 0) circle (.8);
		\fill[color=niceyellow] (-2,2)rectangle (2,-2);
	\end{scope}
	\begin{scope}
		\clip (-0.5, 0) circle (.8);
		\clip ( -1.2, .6) circle (.8);
		\fill[color=niceyellow] (-2,2)rectangle (2,-2);
	\end{scope}
	\begin{scope}
		\clip ( -1, 0) circle (.8);
		\clip ( 0, 1.1) circle (.8);
		\fill[color=niceyellow] (-2,2)rectangle (2,-2);
	\end{scope}
	\begin{scope}
		\clip ( -1.2, .6) circle (.8);
		\clip ( 0, 1.1) circle (.8);
		\fill[color=niceyellow] (-2,2)rectangle (2,-2);
	\end{scope}
	\begin{scope}
		\clip (-0.5, 0) circle (.8);
		\clip ( 0, 1.1) circle (.8);
		\fill[color=niceyellow] (-2,2)rectangle (2,-2);
	\end{scope}
	\begin{scope}
		\clip ( 0.5, .3) circle (.8);
		\clip ( 0, 1.1) circle (.8);
		\fill[color=niceyellow] (-2,2)rectangle (2,-2);
	\end{scope}
	\begin{scope}
		\clip ( -1.2, .6) circle (.8);
		\clip  (-0.5, 0) circle (.8);
		\fill[color=niceyellow] (-2,2)rectangle (2,-2);
	\end{scope}
	\begin{scope}
		\clip  (-0.5, 0) circle (.8);
		\clip  ( -.7, -.6) circle (.8);
		\fill[color=niceyellow] (-2,2)rectangle (2,-2);
	\end{scope}
	\begin{scope}
		\clip ( -1, 0) circle (.8);
		\clip  ( -.7, -.6) circle (.8);
		\fill[color=niceyellow] (-2,2)rectangle (2,-2);
	\end{scope}
	\begin{scope}
		\clip    (-0.5, 0) circle (.8);
		\clip  ( .97, -.50) circle (.8);
		\fill[color=niceyellow] (-2,2)rectangle (2,-2);
	\end{scope}
	\begin{scope}
		\clip  ( 0.5, .3) circle (.8);
		\clip  ( .97, -.50) circle (.8);
		\fill[color=niceyellow] (-2,2)rectangle (2,-2);
	\end{scope}
	\draw (-0.5, 0) circle (.8);
	\node (a)[point] at (-.5,0) {};
	\draw ( 0.5, .3) circle (.8);
	\node (b)[point] at (.5,.3) {};
	\draw ( -1, 0) circle (.8);
	\node (c)[point] at (-1,0) {};
	\draw ( -1.2, .6) circle (.8);
	\node (c)[point] at ( -1.2, .6) {};
	\draw ( 0, 1.1) circle (.8);
	\node (c)[point] at (  0, 1.1) {};	
	\draw ( -.7, -.6) circle (.8);
	\node (c)[point] at (-.7, -.6) {};
	\draw ( .97, -.50) circle (.8);
	\node (c)[point] at ( .97, -.5) {};
\end{tikzpicture}
\begin{tikzpicture}[scale=1.4]
	\tikzstyle{point}=[circle,thick,draw=black,fill=black,inner sep=0pt,minimum width=4pt,minimum height=4pt]
	\draw (-2.2, 2) rectangle (2, -1.6);
	\begin{scope}
		\clip (-0.5, 0) circle (.8);
		\clip ( 0.5, .3) circle (.8);
		\clip ( -1, 0) circle (.8);
		\fill[color=niceyellow] (-2,2)rectangle (2,-2);
	\end{scope}
	\begin{scope}
		\clip (-0.5, 0) circle (.8);
		\clip ( -1.2, .6) circle (.8);
		\clip ( -1, 0) circle (.8);
		\fill[color=niceyellow] (-2,2)rectangle (2,-2);
	\end{scope}
	\begin{scope}
		\clip ( -1, 0) circle (.8);
		\clip ( -1.2, .6) circle (.8);
		\clip ( 0, 1.1) circle (.8);
		\fill[color=niceyellow] (-2,2)rectangle (2,-2);
	\end{scope}
	\begin{scope}
		\clip (-0.5, 0) circle (.8);
		\clip ( 0.5, .3) circle (.8);
		\clip ( 0, 1.1) circle (.8);
		\fill[color=niceyellow] (-2,2)rectangle (2,-2);
	\end{scope}
	\begin{scope}
		\clip ( 0, 1.1) circle (.8);
		\clip ( -1, 0) circle (.8);
		\clip  (-0.5, 0) circle (.8);
		\fill[color=niceyellow] (-2,2)rectangle (2,-2);
	\end{scope}
	\begin{scope}
		\clip ( 0, 1.1) circle (.8);
		\clip ( -1.2, .6) circle (.8);
		\clip  (-0.5, 0) circle (.8);
		\fill[color=niceyellow] (-2,2)rectangle (2,-2);
	\end{scope}
	\begin{scope}
		\clip  (-0.5, 0) circle (.8);
		\clip ( -1, 0) circle (.8);
		\clip  ( -.7, -.6) circle (.8);
		\fill[color=niceyellow] (-2,2)rectangle (2,-2);
	\end{scope}
	\begin{scope}
		\clip    (-0.5, 0) circle (.8);
		\clip  ( 0.5, .3) circle (.8);
		\clip  ( .97, -.50) circle (.8);
		\fill[color=niceyellow] (-2,2)rectangle (2,-2);
	\end{scope}
	\draw (-0.5, 0) circle (.8);
	\node (a)[point] at (-.5,0) {};
	\draw ( 0.5, .3) circle (.8);
	\node (b)[point] at (.5,.3) {};
	\draw ( -1, 0) circle (.8);
	\node (c)[point] at (-1,0) {};
	\draw ( -1.2, .6) circle (.8);
	\node (c)[point] at ( -1.2, .6) {};
	\draw ( 0, 1.1) circle (.8);
	\node (c)[point] at (  0, 1.1) {};	
	\draw ( -.7, -.6) circle (.8);
	\node (c)[point] at (-.7, -.6) {};
	\draw ( .97, -.50) circle (.8);
	\node (c)[point] at ( .97, -.5) {};
\end{tikzpicture}
\caption{The 2- and 3-fold cover of a few points with respect to a certain radius. The first homology of the 2-fold cover is trivial, while the first homology of the 3-fold cover is non-trivial.}
\label{fig:kfold}
\end{figure}
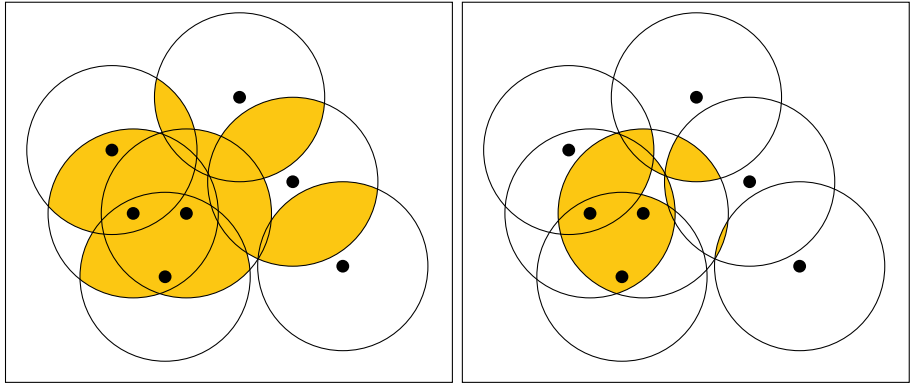

\subsection{Motivation and prior work}

For $k=1$ fixed, $(\cov{r}{1})_{r\in[0,\infty)}$ is the well-known \emph{offset filtration} (also known as the union of balls filtration),
a standard construction for analyzing the topology of a finite point sample across scales~\cite{edelsbrunner2010computational}.  It is a central object in the field of persistent homology.
While the persistent homology of this filtration is stable to small geometric perturbations of the
sites~\cite{Cohen-SteinerEdelsbrunnerHarer2007}, it is not robust with respect to outliers, and it can be insensitive to topological structure in high density regions of the data.

Within the framework of 1-parameter persistent homology, there have been many proposals for alternative constructions which address these issues.  These approaches include the removal of low density outliers~\cite{carlsson2008local}, filtering by a density function~\cite{bobrowski2017topological,chazal2009analysis,chazal2013clustering}, distance to measure constructions~\cite{anai2020dtm,Buchet15,Chazal2011,guibas2013witnessed}, kernel density functions~\cite{phillips2015geometric}, and subsampling~\cite{blumberg2014robust}.  A detailed overview of these approaches can be found in \cite{blumberg2020stability}.  

Several of these constructions have good stability properties or good asymp\-totic behavior.  
However, as explained in \cite{blumberg2020stability}, all of the known 1-parameter persistence strategies for handling outliers or variations in density share certain disadvantages: First, they all depend on a choice of a parameter.  Typically, this parameter specifies a fixed spatial scale or a density threshold at which the construction is carried out.  In the absence of a priori knowledge about the structure of the data, it may be unclear how to select such a parameter.  And if the data exhibits topological features at multiple spatial or density scales, it may be that no single parameter choice allows us to capture all the structure present in the data.  Second, constructions that fix a scale parameter are unable to distinguish between small spatial features and large ones, and constructions that fix a density or measure parameter are unable to distinguish features in densely sampled regions of the data from features involving sparse~regions.  

A natural way to circumvent these limitations is to consider a 2-parameter approach, where one constructs a  bifiltration from the data, rather than a 1-parameter filtration \cite{carlsson2009theory}.  The multicover bifiltration is one natural option for this.  Alternatives include the density bifiltrations of Carlsson and Zomoro\-dian~\cite{carlsson2009theory}, and the degree bifiltrations of Lesnick and Wright  \cite{Lesnick2015}; again, we refer the reader to \cite{blumberg2020stability} for a more detailed discussion.
Among these three options, the multicover bifiltration has two attractive features which together distinguish it from the others.  First, its construction does not depend on any additional parameters. Second, the multicover bifiltration satisfies a strong stability property, which in particular guarantees robustness to outliers \cite{blumberg2020stability}.

There is a substantial and growing literature on the use of bifiltrations in data analysis.  Most approaches begin by applying homology with coefficients to the bifiltration, to obtain an algebraic object called  \emph{bipersistence module}.  In contrast to the 1-parameter case, where the algebraic structure of a persistence module is completely described by a barcode \cite{ZomorodianCarlsson2005}, it is well-known that defining the barcode of bipersistence modules is problematic \cite{carlsson2009theory}.  Nevertheless, one can compute invariants of a bipersistence module which serve as useful surrogates for a barcode, and a number of ideas for this have been proposed  \cite{carlsson2009theory,cerri2013betti,harrington2017stratifying,Lesnick2015,vipond2018multiparameter}.

Regardless of which invariants of the multicover bifiltration we wish to consider, to work with this bifiltration computationally, the natural first step is to find a reasonably sized combinatorial  (i.e., simplicial or polyhedral) model for the bifiltration.  With such a model, recently developed algorithms such as those described in \cite{lesnick2019computing} and \cite{kerber2021fast} can efficiently compute minimal presentations and standard invariants of the homology modules of the bifiltration.

In the 1-parameter setting, there are two well-known simplicial models of the offset filtration. The \emph{\v{C}ech filtration}, is given at each scale by the nerve of the balls; the equivalence of the offset and \v Cech filtrations follows from the \emph{Persistent Nerve Theorem}.  For large point sets, the full \v Cech filtration is too large to be used in practical computations.  However, the \emph{alpha filtration} (also known as the \emph{Delaunay filtration})~\cite{Edelsbrunner1995,edelsbrunner2010computational} is a much smaller subfiltration of the \v Cech filtration which is also simplicial model for the offset filtration.  It is given at each scale by intersecting each ball with the Voronoi cell~\cite{voronoi1908recherches} of its center, and then taking the nerve of the resulting regions.  For $d$ small (say $d\leq 3$), the Delaunay filtration is readily computed in practice for many thousands of~points.

It is implicit in the work of Sheehy~\cite{Sheehy2012} that the multicover bifiltration has an elegant simplicial model, the \emph{subdivision-\v Cech bifiltration}, obtained via a natural filtration on the barycentric subdivision of each \v Cech complex; see also \cite[Appendix B]{cavanna17when} and \cite[Section 4]{blumberg2020stability}.  However, the subdivision-\v Cech bifiltration has exponentially many vertices in the size of the data, making it even more unsuitable for computations than the ordinary \v{C}ech filtration.

Edelsbrunner and Osang~\cite{Edelsbrunner2018} therefore seek to develop the computational theory of the multicover bifiltration using \emph{higher-order Delaunay complexes}~\cite{Edelsbrunner_shape_2006,Krasnoshchekov},  taking the alpha filtration as inspiration.  Assuming the sites are in general position, they define a polyhedral cell complex in $\R^{d+1}$ called the \emph{rhomboid tiling}, which contains all higher-order Delaunay complexes as planar~sections. 
Using the rhomboid tiling, they present a polynomial time algorithm to \linebreak compute the barcodes of a horizontal or vertical slice of the multicover bi\-filtration (i.e., of a one-parameter filtration obtained by fixing either one of the two parameters $r,k$). 
The case of fixed $r$ and varying $k$ is more challenging because the order-$k$ Delaunay complexes do not form a filtration.  The authors construct a zigzag filtration for this case. 
The problem of efficiently computing 2-para\-meter persistent homology of the multicover bifiltration is not addressed by~\cite{Edelsbrunner2018}.

We note that the subdivision-\v Cech bifiltration is more general than the Rhomboid tiling:  the rhomboid tiling is defined only for Euclidean data, \linebreak whereas the topological equivalence of the subdivision-\v Cech and multicover bifiltrations extends to data in any metric space where finite intersections of balls are contractible.

\subsection{Contributions}
We introduce the first efficiently computable combinatorial models of the 
 multicover bifiltration $\covonly$.  First, we introduce a simplicial model, whose
construction is based on two main ideas: In order to connect
the higher-order alpha complex constructions for $(r,k)$ and $(r,k+1)$, we simply overlay
their underlying covers to a ``double-cover'', whose nerve is
a simplicial complex that contains both alpha complexes. This yields a zigzag of simplicial filtrations.  The second main idea is that this zigzag can be ``straightened out'' to a (non-zigzaging) bifiltration,
simply by taking unions of prefixes in the zigzag sequence.  This straightening technique has previously been used by Sheehy to construct sparse approximations of 
Vietoris-Rips complexes~\cite{Sheehy2013}.  Together, these two ideas give rise to a bifiltration $\tcovonly$ 
of simplicial complexes.  

The bifiltration $\tcovonly$ can also be obtained directly as the persistent nerve of a ``thickening" of $\covonly$ constructed via mapping telescopes.  This observation leads to a simple 
proof of topological equivalence (i.e., weak equivalence; see Section~\ref{sec:background}) of $\covonly$ and $\tcovonly$ via the Nerve Theorem.  It follows that the persistent homology modules of $\covonly$ and $\tcovonly$ are isomorphic.

Our second contribution is to show that the rhomboid tiling
as defined in~\cite{Edelsbrunner2018} also gives rise to a (non-zigzaging) bifiltration
of polyhedral complexes that is topologically equivalent to the multicover.
We proceed in two steps: 
First, we slice every rhomboid at each integer value $k$
(slightly increasing the number of cells) and adapt the straightening trick used to construct $\tcovonly$.  
We prove the topological equivalence of this construction 
with the multicover bifiltration by relating the slice rhomboid filtration with $\tcovonly$.
The main observation is a one-to-one correspondence of maximal-dimensional cells
in both constructions, which leads to a proof via the Nerve Theorem.
Second, we relate the sliced and unsliced rhomboid tilings at every scale via a deformation retraction.

We give size bounds for both of the bifiltrations we introduce. For $n$ points in $\R^d$, we show that their size is $O(n^{d+1})$.  This is a decisive improvement over Sheehy's \v{C}ech-based construction, which has exponential dependence on $n$.

An efficient algorithm for computing rhomboid tilings has recently been presented in~\cite{EdOs20};
hence, using the accompanying implementation \linebreak \textsc{rhomboidtiling} of this algorithm, 
our second contribution gives us an efficient software to compute 
a bifiltration of cell complexes equivalent to $\covonly$, currently for points in $\R^2$ and $\R^3$.
We combine this implementation with the libraries \textsc{mpfree}
and \textsc{rivet} to demonstrate that minimal presentations of multicover persistent homology modules can now be efficiently computed, often within seconds, as can invariants such as the Hilbert function.

\section{Background}
\label{sec:background}

\subsection{Filtrations} 

For $P$ a poset, define a ($P$-indexed) filtration to be a collection of topological spaces $X=(X_p)_{p\in P}$ indexed by $P$, such that $X_p\subseteq X_q$ whenever $p\leq q\in P$.  For example, an $\N$-indexed filtration $X$ is a diagram of spaces and inclusions of the following form:
\[
\begin{tikzcd}
X_{0}  \arrow[hookrightarrow]{r} & X_{1} \arrow[hookrightarrow]{r} & X_{2}  \arrow[hookrightarrow]{r}&  \cdots.
\end{tikzcd}
\]
A \emph{morphism} $\varphi:X\to Y$ of $P$-indexed filtrations is a collection of continuous maps $(\varphi_p:X_p\to Y_p)_{p\in P}$ which commute with the inclusions in $X$ and $Y$.  In the language of category theory, a $P$-indexed filtration is a functor $P\to \mathbf{Top}$ whose internal maps are inclusions, and a morphism is a natural transformation.

Recall that the product poset $P\times Q$ of posets $P$ and $Q$ is defined by taking $(p,q)\leq (p',q')$ if and only if $p\leq p'$ and $q\leq q'$.    
When $P$ is the product of two totally ordered sets, we call a $P$-indexed filtration a \emph{bifiltration}.  

In the classical homotopy theory of diagrams of spaces, there is a standard analogue of the notion of homotopy equivalence for diagrams of spaces, called \emph{weak equivalence}.  We now define a version of this for $P$-indexed filtrations: A morphism of filtrations $\varphi:X\to Y$ is called an \emph{objectwise homotopy equivalence} if each $\varphi_p:X_p\to Y_p$ is a homotopy equivalence. If there exists a finite zigzag diagram of objectwise homotopy equivalences 
\[X\rightarrow Z_1 \leftarrow Z_2 \rightarrow \cdots \leftarrow Z_{n-1}\rightarrow Z_n \leftarrow Y\] 
connecting $X$ and $Y$,  then we say that $X$ and $Y$ are \emph{weakly equivalent}.  The terminology originates from the theory of model categories \cite{dwyer1995homotopy,hirschhorn2009model}.  See \cite{blumberg2017universality,lanari2020rectification,scoccola2020locally} for discussions of weak equivalence of diagrams in the context of TDA.
 
\begin{remark}
To motivate the consideration of zigzags in the definition above, we note that for $X$ and $Y$ a pair of weakly equivalent $P$-indexed filtrations, there is not necessarily an objectwise homotopy equivalence $f:X\to Y$.  For example, let $P=\R$, $X$ be the offset filtration on $\{0,1\}\subset \R$, and $Y$ be the nerve filtration of $X$.  It is easy to check that there is no objectwise homotopy equivalence $f:X\to Y$.  On the other hand, it follows from the Persistent Nerve Theorem (Theorem~\ref{thm:persistentnerve} below) that $X$ and $Y$ are weakly equivalent.  Moreover, one can construct a similar example of weakly equivalent filtrations for which there is no objectwise homotopy equivalence in either direction.  
\end{remark}
 
An objectwise homotopy equivalence $\varphi:X\to Y$ induces isomorphisms on the persistent homology modules of $X$ and $Y$.  Hence, weakly equivalent filtrations have isomorphic persistent homology modules.

We say a $P$-indexed filtration $X$ is \emph{Euclidean} if $X_p\subset \R^n$ for some $n$ and all $p\in P$.

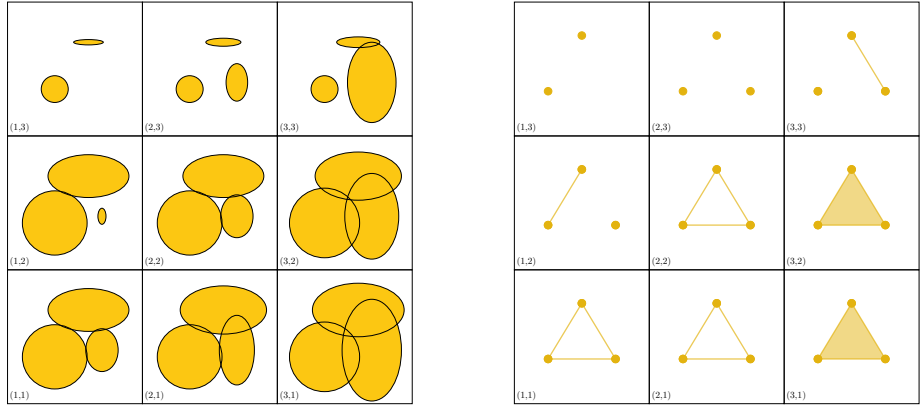
\begin{figure}
\centering
\begin{minipage}[t]{0.45\textwidth}
\scalebox{.355}{
\begin{tikzpicture}
 \tikzstyle{point}=[circle,thick,draw=black!10!niceyellow,fill=black!10!niceyellow,inner sep=0pt,minimum width=4pt,minimum height=4pt]
  \tikzset{
        basefont/.style = {font = \Large\sffamily},
          timing/.style = {basefont, sloped,above},
           label/.style = {basefont, align = left},
          screen/.style = {basefont, black, align = center,
                           minimum size = 6cm, fill = black!6, draw = black}};
\def\scale{5}
\foreach \x in {1,2,3} {
    \foreach \y in {1,2,3} {
        \node at (\scale*\x-\scale+0.45,\scale*\y-\scale+0.3) {(\x,\y)};
        \path[draw,thick] (\scale*\x-\scale,\scale*\y-\scale) rectangle (\scale*\x,\scale*\y);
    }
}
\def\Ax{3}
\def\Ay{3.5}
\def\Bx{1.75}
\def\By{1.75}
\def\Cx{3.5}
\def\Cy{2}


	\def\firstcov{(\Ax,\scale+\scale+\Ay) ellipse (0.55cm and 0.1cm)}
	\def\secondcov{(\Bx,\scale+\scale+\By) circle (0.5cm)}
	\def\thirdcov{(\Cx,\scale+\scale+\Cy) ellipse (.05cm and .05cm)}
	\fill[niceyellow] \firstcov;
	\fill[niceyellow] \secondcov;
    \draw \firstcov node(1)[below] {};
   \draw \secondcov node (2) [above] {};

	\def\firstcov{(\scale+\Ax,\scale+\scale+\Ay) ellipse (.65cm and .15cm)}
	\def\secondcov{(\scale+\Bx,\scale+\scale+\By) circle (.5cm)}
	\def\thirdcov{(\scale+\Cx,\scale+\scale+\Cy) ellipse (.4cm and .7cm)}
	\fill[niceyellow] \firstcov;
	\fill[niceyellow] \secondcov;
	\fill[niceyellow] \thirdcov;
    \draw \firstcov node[below] {};
    \draw \secondcov node [above] {};
    \draw \thirdcov node [below] {};

	\def\firstcov{(\scale+\scale+\Ax,\scale+\scale+\Ay) ellipse (.8cm and .2cm)}
	\def\secondcov{(\scale+\scale+\Bx,\scale+\scale+\By) circle (.5cm)}
	\def\thirdcov{(\scale+\scale+\Cx,\scale+\scale+\Cy) ellipse (.9cm and 1.5cm)}
	\fill[niceyellow] \firstcov;
	\fill[niceyellow] \secondcov;
	\fill[niceyellow] \thirdcov;
    \draw \firstcov node[below] {};
    \draw \secondcov node [above] {};
    \draw \thirdcov node [below] {};

	\def\firstcov{(\Ax,\scale+\Ay) ellipse (1.5cm and .8cm)}
	\def\secondcov{(\Bx,\scale+\By) circle (1.2cm)}
	\def\thirdcov{(\Cx,\scale+\Cy) ellipse (.15cm and .3cm)}
	\fill[niceyellow] \firstcov;
	\fill[niceyellow] \secondcov;
	\fill[niceyellow] \thirdcov;
    \draw \firstcov node[below] {};
    \draw \secondcov node [above] {};
    \draw \thirdcov node [below] {};

	\def\firstcov{(\scale+\Ax,\scale+\Ay) ellipse (1.5cm and .8cm)}
	\def\secondcov{(\scale+\Bx,\scale+\By) circle (1.2cm)}
	\def\thirdcov{(\scale+\Cx,\scale+\Cy) ellipse (.6cm and .8cm)}
	\fill[niceyellow] \firstcov;
	\fill[niceyellow] \secondcov;
	\fill[niceyellow] \thirdcov;
    \draw \firstcov node[below] {};
    \draw \secondcov node [above] {};
    \draw \thirdcov node [below] {};

	\def\firstcov{(2*\scale+\Ax,\scale+\Ay) ellipse (1.6cm and .9cm)}
	\def\secondcov{(2*\scale+\Bx,\scale+\By) circle (1.3cm)}
	\def\thirdcov{(2*\scale+\Cx,\scale+\Cy) ellipse (1cm and 1.6cm)}
	\fill[niceyellow] \firstcov;
	\fill[niceyellow] \secondcov;
	\fill[niceyellow] \thirdcov;
    \draw \firstcov node[below] {};
    \draw \secondcov node [above] {};
    \draw \thirdcov node [below] {};

	\def\firstcov{(\Ax,\Ay) ellipse (1.5cm and .8cm)}
	\def\secondcov{(\Bx,\By) circle (1.2cm)}
	\def\thirdcov{(\Cx,\Cy) ellipse (.6cm and .8cm)}
	\fill[niceyellow] \firstcov;
	\fill[niceyellow] \secondcov;
	\fill[niceyellow] \thirdcov;
    \draw \firstcov node[below] {};
    \draw \secondcov node [above] {};
    \draw \thirdcov node [below] {};

	\def\firstcov{(\scale+\Ax,\Ay) ellipse (1.6cm and .9cm)}
	\def\secondcov{(\scale+\Bx,\By) circle (1.2cm)}
	\def\thirdcov{(\scale+\Cx,\Cy) ellipse (.65cm and 1.3cm)}
	\fill[niceyellow] \firstcov;
	\fill[niceyellow] \secondcov;
	\fill[niceyellow] \thirdcov;
    \draw \firstcov node[below] {};
    \draw \secondcov node [above] {};
    \draw \thirdcov node [below] {};

	\def\firstcov{(2*\scale+\Ax,\Ay) ellipse (1.7cm and 1cm)}
	\def\secondcov{(2*\scale+\Bx,\By) circle (1.3cm)}
	\def\thirdcov{(2*\scale+\Cx,\Cy) ellipse (1.1cm and 1.9cm)}
	\fill[niceyellow] \firstcov;
	\fill[niceyellow] \secondcov;
	\fill[niceyellow] \thirdcov;
    \draw \firstcov node[below] {};
    \draw \secondcov node [above] {};
    \draw \thirdcov node [below] {};

\end{tikzpicture}
}
\end{minipage}
 \, \, \, \, \,\, \,
\begin{minipage}[t]{0.45\textwidth}
\scalebox{.355}{
\begin{tikzpicture}
 \tikzstyle{point}=[circle,thick,draw=black!10!niceyellow,fill=black!10!niceyellow,inner sep=0pt,minimum width=8pt,minimum height=8pt]
  \tikzset{
        basefont/.style = {font = \Large\sffamily},
          timing/.style = {basefont, sloped,above},
           label/.style = {basefont, align = left},
          screen/.style = {basefont, black, align = center,
                           minimum size = 6cm, fill = black!6, draw = black}};
\def\scale{5}
\def\Bx{1.25}
\def\By{1.67}
\def\Ax{2.5}
\def\Ay{3.75}
\def\Cx{3.75}
\def\Cy{1.67}

\foreach \x in {1,2,3} {
    \foreach \y in {1,2,3} {
        \node at (\scale*\x-\scale+0.45,\scale*\y-\scale+0.3) {(\x,\y)};    
        \path[draw,thick] (\scale*\x-\scale,\scale*\y-\scale) rectangle (\scale*\x,\scale*\y);
    }
}


   	\node (1) at (\Ax,\scale+\scale+\Ay) [point] {};
   	\node (2) at (\Bx,\scale+\scale+\By) [point] {};
   	\node (1) at (\Ax+\scale,\scale+\scale+\Ay) [point] {};
   	\node (2) at (\Bx+\scale,\scale+\scale+\By) [point] {};
   	\node (3) at (\Cx+\scale,\scale+\scale+\Cy) [point] {};
   	\node (1) at (\Ax+\scale+\scale,\scale+\scale+\Ay) [point] {};
   	\node (2) at (\Bx+\scale+\scale,\scale+\scale+\By) [point] {};
   	\node (3) at (\Cx+\scale+\scale,\scale+\scale+\Cy) [point] {};
 	 \draw[pattern=north east lines, pattern color=black!10!niceyellow, draw=black!10!niceyellow,line width=.05cm,opacity=.7] (1.center) --  (3.center) -- cycle;
   	\node (1) at (\Ax+\scale+\scale,\scale+\scale+\Ay) [point] {};
   	\node (2) at (\Bx+\scale+\scale,\scale+\scale+\By) [point] {};
   	\node (3) at (\Cx+\scale+\scale,\scale+\scale+\Cy) [point] {};
   	\node (1) at (\Ax,\Ay+\scale) [point] {};
   	\node (2) at (\Bx,\By+\scale) [point] {};
   	\node (3) at (\Cx,\Cy+\scale) [point] {};
 	 \draw[pattern=north east lines, pattern color=black!10!niceyellow, draw=black!10!niceyellow,line width=.05cm,opacity=.7] (1.center) -- (2.center) -- cycle;
   	\node (1) at (\Ax,\Ay+\scale) [point] {};
   	\node (2) at (\Bx,\By+\scale) [point] {};
   	\node (3) at (\Cx,\Cy+\scale) [point] {};
   	\node (1) at (\Ax+\scale,\Ay+\scale) [point] {};
   	\node (2) at (\Bx+\scale,\By+\scale) [point] {};
   	\node (3) at (\Cx+\scale,\Cy+\scale) [point] {};
 	 \draw[pattern color=black!10!niceyellow, draw=black!10!niceyellow,line width=.05cm,opacity=.7] (1.center) -- (2.center) -- cycle;
 	 \draw[pattern=north east lines, pattern color=black!10!niceyellow, draw=black!10!niceyellow,line width=.05cm,opacity=.7] (3.center) -- (2.center) -- cycle;
 	 \draw[pattern=north east lines, pattern color=black!10!niceyellow, draw=black!10!niceyellow,line width=.05cm,opacity=.7] (3.center) -- (1.center) -- cycle;
   	\node (1) at (\Ax+\scale,\Ay+\scale) [point] {};
   	\node (2) at (\Bx+\scale,\By+\scale) [point] {};
   	\node (3) at (\Cx+\scale,\Cy+\scale) [point] {};
   	\node (1) at (\Ax+\scale+\scale,\Ay+\scale) [point] {};
   	\node (2) at (\Bx+\scale+\scale,\By+\scale) [point] {};
   	\node (3) at (\Cx+\scale+\scale,\Cy+\scale) [point] {};
 	 \fill[black!10!niceyellow,opacity=.5] (1.center) -- (2.center) -- (3.center) -- cycle;
 	 \draw[pattern=north east lines, pattern color=black!10!niceyellow, draw=black!10!niceyellow,line width=.05cm,opacity=.7] (1.center) -- (2.center) -- cycle;
 	 \draw[pattern=north east lines, pattern color=black!10!niceyellow, draw=black!10!niceyellow,line width=.05cm,opacity=.7] (3.center) -- (2.center) -- cycle;
 	 \draw[pattern=north east lines, pattern color=black!10!niceyellow, draw=black!10!niceyellow,line width=.05cm,opacity=.7] (3.center) -- (1.center) -- cycle;
   	\node (1) at (\Ax+\scale+\scale,\Ay+\scale) [point] {};
   	\node (2) at (\Bx+\scale+\scale,\By+\scale) [point] {};
   	\node (3) at (\Cx+\scale+\scale,\Cy+\scale) [point] {};
   	\node (1) at (\Ax,\Ay) [point] {};
   	\node (2) at (\Bx,\By) [point] {};
   	\node (3) at (\Cx,\Cy) [point] {};
 	 \draw[pattern=north east lines, pattern color=black!10!niceyellow, draw=black!10!niceyellow,line width=.05cm,opacity=.7] (1.center) -- (2.center) -- cycle;
 	 \draw[pattern=north east lines, pattern color=black!10!niceyellow, draw=black!10!niceyellow,line width=.05cm,opacity=.7] (3.center) -- (2.center) -- cycle;
 	 \draw[pattern=north east lines, pattern color=black!10!niceyellow, draw=black!10!niceyellow,line width=.05cm,opacity=.7] (3.center) -- (1.center) -- cycle;
   	\node (1) at (\Ax,\Ay) [point] {};
   	\node (2) at (\Bx,\By) [point] {};
   	\node (3) at (\Cx,\Cy) [point] {};
   	\node (1) at (\Ax+\scale,\Ay) [point] {};
   	\node (2) at (\Bx+\scale,\By) [point] {};
   	\node (3) at (\Cx+\scale,\Cy) [point] {};
 	 \draw[pattern=north east lines, pattern color=black!10!niceyellow, draw=black!10!niceyellow,line width=.05cm,opacity=.7] (1.center) -- (2.center) -- cycle;
 	 \draw[pattern=north east lines, pattern color=black!10!niceyellow, draw=black!10!niceyellow,line width=.05cm,opacity=.7] (3.center) -- (2.center) -- cycle;
 	 \draw[pattern=north east lines, pattern color=black!10!niceyellow, draw=black!10!niceyellow,line width=.05cm,opacity=.7] (3.center) -- (1.center) -- cycle;
   	\node (1) at (\Ax+\scale,\Ay) [point] {};
   	\node (2) at (\Bx+\scale,\By) [point] {};
   	\node (3) at (\Cx+\scale,\Cy) [point] {};
   	\node (1) at (\Ax+\scale+\scale,\Ay) [point] {};
   	\node (2) at (\Bx+\scale+\scale,\By) [point] {};
   	\node (3) at (\Cx+\scale+\scale,\Cy) [point] {};
 	 \fill[black!10!niceyellow,opacity=.5] (1.center) -- (2.center) -- (3.center) -- cycle;
 	 \draw[pattern=north east lines, pattern color=black!10!niceyellow, draw=black!10!niceyellow,line width=.05cm,opacity=.7] (1.center) -- (2.center) -- cycle;
 	 \draw[pattern=north east lines, pattern color=black!10!niceyellow, draw=black!10!niceyellow,line width=.05cm,opacity=.7] (3.center) -- (2.center) -- cycle;
 	 \draw[pattern=north east lines, pattern color=black!10!niceyellow, draw=black!10!niceyellow,line width=.05cm,opacity=.7] (3.center) -- (1.center) -- cycle;
   	\node (1) at (\Ax+\scale+\scale,\Ay) [point] {};
   	\node (2) at (\Bx+\scale+\scale,\By) [point] {};
   	\node (3) at (\Cx+\scale+\scale,\Cy) [point] {};
\end{tikzpicture}
}
\end{minipage}
\caption{\emph{Left:} A bifiltration of good covers over $\{1<2<3\}\times\{3<2<1\}$. \emph{Right:} A bifiltration consisting of the nerves of the covers. The Persistent Nerve Theorem ensures that not only the individual spaces at scales $(n,m)$ are homotopy equivalent, but also that the two bifiltrations are weakly equivalent.}
\label{fig:persistentnerve}
\end{figure}


\subsection{The Persistent nerve theorem}\label{paragraph:nervethm}
A \emph{cover} of $X\subset \R^n$ is a collection $\firstcovername=\{\firstcovername^i\}_{i\in I}$ of subsets of $X$ whose union is $X$.   
The \emph{nerve} of $\firstcovername$ is the abstract simplicial complex 
\[
\nerve\left(\firstcovername\right):=\left\{\sigma\subset I\mid\bigcap_{i\in \sigma} \firstcovername^i \neq \emptyset \right\}.
\]  
We say that the cover is \emph{good} if it is finite and consists of closed, convex sets~\cite{edelsbrunner2010computational}.

One version of the \emph{Nerve Theorem}  asserts that $X$ and 
$\nerve\left(\firstcovername\right)$ are homotopy equivalent whenever $\firstcovername$ is a good cover of $X$~\cite{edelsbrunner2010computational,leray45}.
It is TDA folklore that this version of the Nerve Theorem can be extended to a persistent version; a proof appears in \cite{kerber2020}; see also \cite{chazal2008towards} for formulation of the Persistent Nerve Theorem in terms of open covers.  

In order to state the Persistent Nerve Theorem for closed, convex covers, we first extend the definition of a cover.  

\begin{definition}[Cover of a filtration]\label{Def:_Cover_of_Filtration}
Let $P$ be a poset and $X$ a $P$-indexed Euclidean filtration.  A \emph{cover of $X$} is a collection 
$\firstcovername=\{\firstcovername^i\}_{i\in I}$ of $P$-indexed filtrations  such that for each $p\in P$, $\{\firstcovername^i_p \mid i\in I\}$ is a cover of $\firstcovername_p$.  We say $\firstcovername$ is good if each $\firstcovername_p:=\{\firstcovername^i_p \mid i\in I\}$ is a good cover.
\end{definition}

The definition of the nerve above extends immediately to yield a \emph{nerve filtration} $\nerve\left(\firstcovername\right)$ associated to any cover of a filtration.
\begin{theorem}[Persistent Nerve Theorem~\cite{kerber2020}]\label{thm:persistentnerve}
Let $P$ be a poset, $X$ a $P$-indexed Euclidean filtration, and $\firstcovername$ a good cover of $X$.  There exists a diagram of objectwise homotopy equivalences
\[
\begin{tikzcd}
 X  &  \arrow[labels=above]{l}{\simeq} \arrow{r}{\simeq} \deltaspace{\firstcovername}{X}  & \nerve(\firstcovername).
 \end{tikzcd}
\]
\end{theorem} 
As shown in \cite{kerber2020}, the intermediate filtration $\deltaspace{\firstcovername}{X}$ in the statement of the theorem can be taken to be a homotopy colimit of a diagram constructed from $\firstcovername$, just as in the proof of the Persistent Nerve Theorem for open covers \cite{chazal2008towards}.  Note that if $P$ is a singleton set, the Persistent Nerve Theorem specializes to the classical version of the Nerve Theorem mentioned~above.

\begin{figure}
\centering
\includegraphics[width=0.495\textwidth]{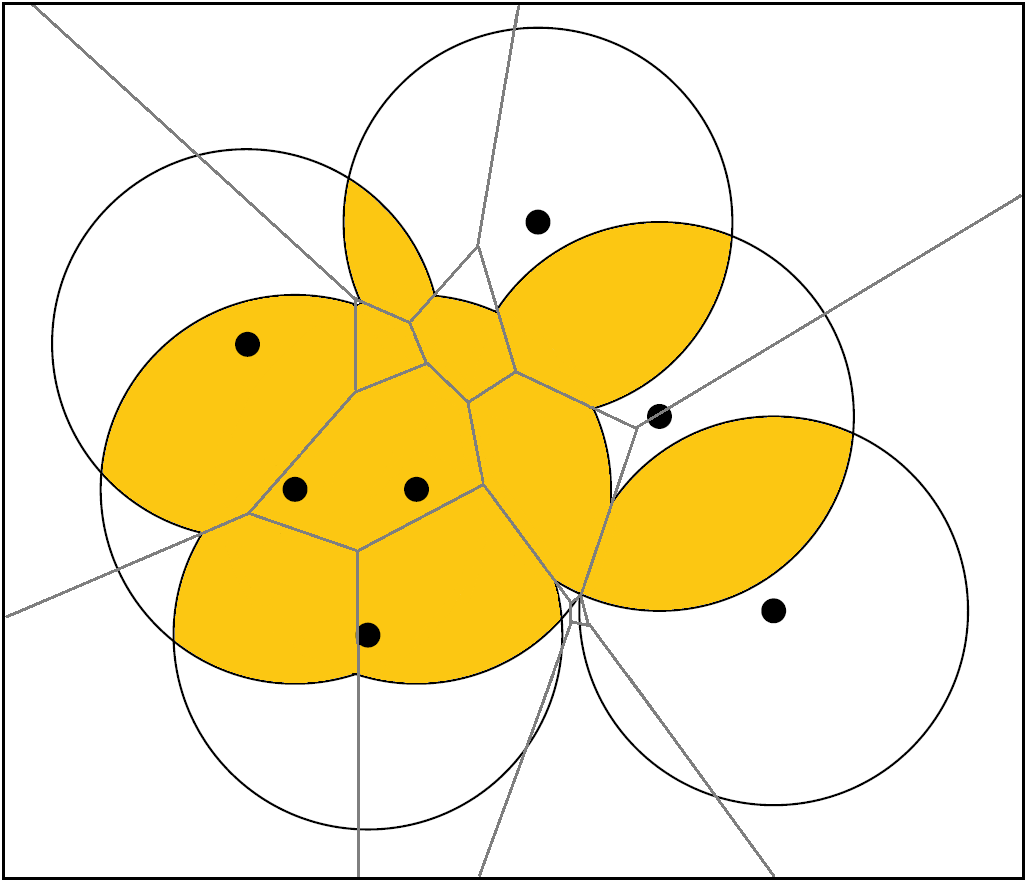}
\includegraphics[width=0.495\textwidth]{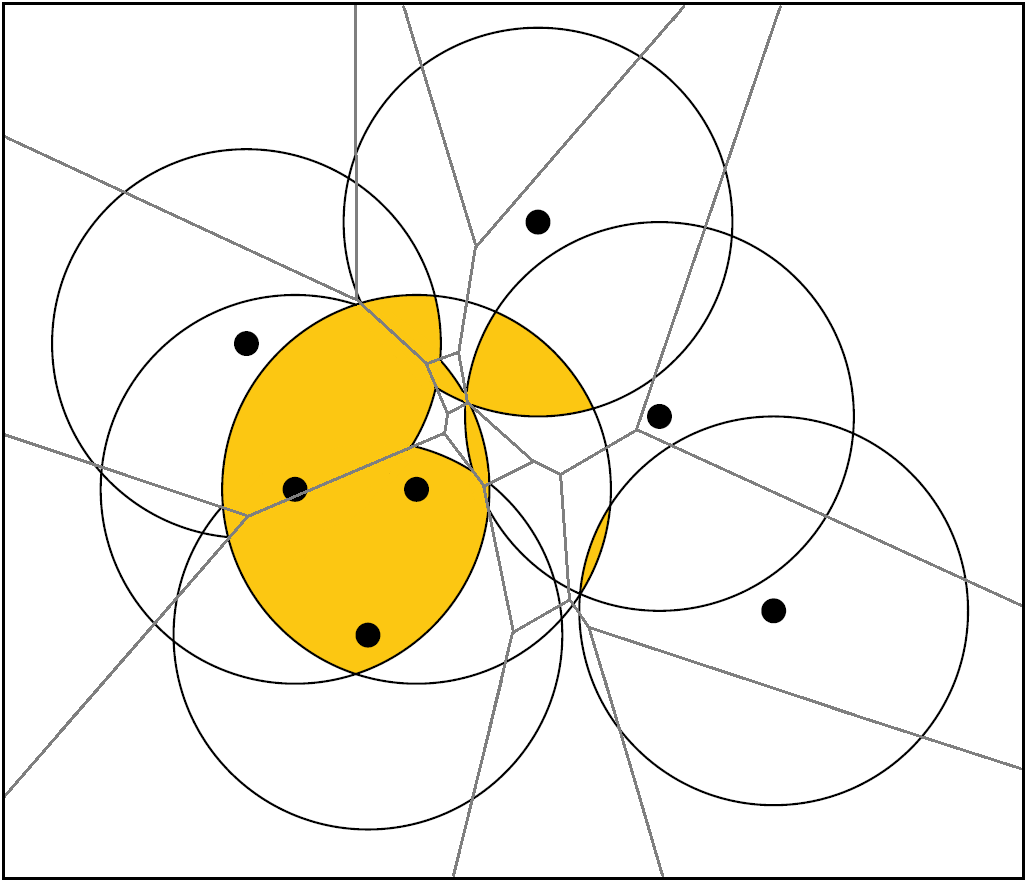}
\caption{\emph{Left:} The 2-fold cover of some points with respect to a certain radius overlapped with its Voronoi diagram of order 2. $\vorcovonly$ is combinatorially simpler than  
$\bigcovonly$.
\emph{Right:} The corresponding 3-fold cover overlapped with its Voronoi diagram of order 3.}
\label{fig:vorcov}
\end{figure} 

\subsection{Multicovers} 

As indicated in the introduction, the multicover bifiltration of a finite set $A\subset \R^d$ is the $\R\times \N^{\mathrm{op}}$-indexed bifiltration $\covonly$ given by
\[ 
\cov{r}{k}:=\left\{  b\in\mathbb{R}^d \mid \textnormal{ } ||b-a||\leq r \textnormal{ for at least $k$ points } a\in A  \right\},
\]
where as above, $\N:=\{0,1,2,\dots\}$.  Note that $\cov{r}{0}=\R^d$ for all $r\in [0,\infty)$.  Given this, one may wonder why we allow $k$ to take the value $0$ in our definition of $\covonly$.  As we will see in Section \ref{sec:rhomboid}, this turns out to be convenient for comparing $\covonly$ to the rhomboid bifiltration.  

As a first step towards constructing a simplicial model of $\covonly$, we identify a good cover for $\cov{r}{k}$ for fixed $r$ and $k$.
For $\widetilde{A}\subset A$ we define
\[ \covreg{r}{\widetilde{A}}:=\left\{ b\in\mathbb{R}^d \mid\textnormal{ } ||b-\tilde{a}||\leq r \textnormal{ for all } \tilde{a}\in \widetilde{A} \right\}.
\]
Each $\covreg{r}{\widetilde{A}}$ is closed and convex.  Letting \[\bigcov{r}{k}:=\left\{ \covreg{r}{\widetilde{A}} \mid  \widetilde{A}\subset A,  |\widetilde{A}|=k \right\},\] we have that $\bigcov{r}{k}$ is a good cover of $\cov{r}{k}$.  
Hence by the Nerve Theorem, $\cov{r}{k}$ is homotopy equivalent to $\nerve(\bigcov{r}{k})$.  For fixed $k$ and $r\leq r'$, we have an inclusion $\nerve(\bigcov{r}{k})\hookrightarrow \nerve(\bigcov{r'}{k})$.  

Note that these nerves can be quite large: for large enough $r$, $\nerve(\bigcov{r}{k})$ contains $\binom{|A|}{k}$ vertices.   To obtain a smaller simplicial model of $\cov{r}{k}$, we use the generalization of Delaunay triangulations to higher-order Delaunay complexes. For a subset $\tilde{A} \subset A$ with $|\tilde{A}|=k$, define its \emph{order-$k$ Voronoi region} as 
\[
\vorreg{\tilde{A}}:=\left\{ b\in \mathbb{R}^d \mid \ ||b-\tilde{a}||\leq ||b-a|| \textnormal{ for all } \tilde{a}\in \tilde{A}, a\in A \setminus \tilde{A} \right\}.
\]
The set of all order-$k$ Voronoi regions  yield a decomposition of $\R^d$ into closed convex subsets having the same $k$ closest points of $A$ in common.  This decom\-position is called the \emph{order-$k$ Voronoi diagram}~\cite{aurenhammer1987power,Edelsbrunner1986}.  We denote it as $\vor{k}$.

For any $r\in \R$ and $k\in \N$, intersecting each order-$k$ Voronoi region with the corresponding multicovered region of fixed~radius~$r$ yields the following good cover of $\cov{r}{k}$.
\[
\vorcov{r}{k}:=\left\{ \covreg{r}{\widetilde{A}}\cap \vorreg{\tilde{A}}\mid\tilde{A}\subset A, |\tilde{A}|=k \right\}.
\]
For an illustration, see Figure~\ref{fig:vorcov}. 
The nerve of $\vorcov{r}{k}$, which we will denote $\del{r}{k}$, is called an \emph{order-$k$ Delaunay complex}.  
By the Nerve Theorem, $\del{r}{k}$ and $\cov{r}{k}$ are homotopy equivalent. 
Note that $\del{r}{1}$  is the \emph {alpha complex} of radius $r$~\cite{Edelsbrunner1995,edelsbrunner2010computational}, whereas $\del{r}{0}$  is a single point for all $r\in [0,\infty)$.
A different but related  concept is the \emph{order-$k$ Delaunay mosaic}, which is the geometric dual of the order-$k$ Voronoi diagram~\cite{Edelsbrunner2018}; see Section~\ref{sec:rhomboid}.
 

\section{A simplicial Delaunay bifiltration}
\label{sec:bifiltration}

\begin{figure}
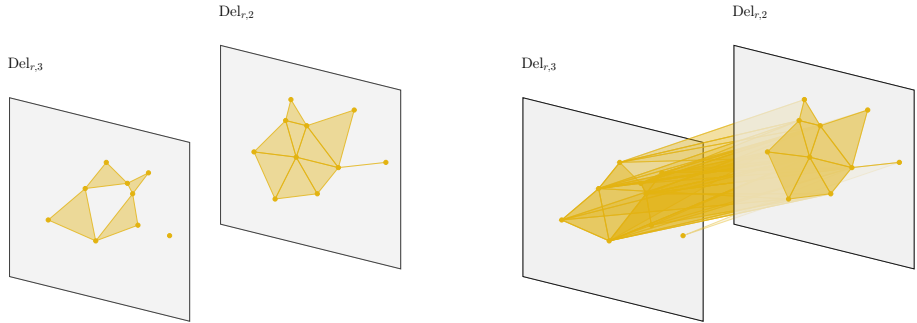

\centering
\scalebox{.395}{

}
\caption{\emph{Left:} The Delaunay complexes of order 2 and 3 of our running example. \emph{Right:} The construction of the simplicial complex $\tdel{r}{2}$. $\tdel{r}{2}$ consists of the Delaunay complexes ${\del{r}{2}}$ and ${\del{r}{3}}$, and additional mixed simplices connecting these. This connection arises from intersections of the $2$-, and $3$-fold covers restricted to their Voronoi diagrams of order 2 and 3, respectively.}
\label{fig:delconstruction}
\end{figure}

For fixed $r\geq 0$, we have inclusions
\[ 
\begin{tikzcd}
\cdots \arrow[hookrightarrow]{r} & \cov{r}{2} \arrow[hookrightarrow]{r} & \cov{r}{1} \arrow[hookrightarrow]{r} & \cov{r}{0},
\end{tikzcd}
\]

but there are no analogous inclusions  $\del{r}{k}\hookrightarrow \del{r}{k-1}$ between the higher-order Delaunay complexes. 
Indeed, we do not even have inclusions of the vertex sets.  Consequently, $(\vorcov{r}{k})_{(r,k)\in [0,\infty)\times\N}$ is not a bifiltered cover of $\covonly$.  

To correct for this, we first define
\[ 
 \tdel{r}{k}:=\nerve\left( \vorcov{r}{k} \cup \vorcov{r}{k+1} \right).
\]
See Figure~\ref{fig:delconstruction} for an illustration.  Note that $\vorcov{r}{k} \cup \vorcov{r}{k+1}$ covers the same space as  $\vorcov{r}{k}$ and that we have inclusions
\[
\begin{tikzcd}
  \del{r}{k+1}  \arrow[hookrightarrow]{r} & \tdel{r}{k} \arrow[hookleftarrow]{r} &  \del{r}{k}.
\end{tikzcd}  
\]
Letting $n:=|A|$, we see that $\tdel{r}{n}=\del{r}{n}$ and $\tdel{r}{k}=\emptyset$ for all $k>n$. 

We now define a $[0,\infty)\times \N^{\mathrm{op}}$-indexed simplicial bifiltration $\tcovonly$ by
\[
\tcov{r}{k}:= \bigcup_{i\geq k} \tdel{r}{i}=\bigcup_{i= k}^n \tdel{r}{i}.
\]
For an illustration, see Figure~\ref{fig:excisionidea}. 
 Note that $\tcov{r}{k}$ is generally \emph{not} equal to the nerve of the union of all $\vorcov{r}{i}$, $i\geq k$, which is a much larger object. 
 
\begin{figure}
\centering
\scalebox{.335}{

}
\caption{
$\tcov{r}{2}$ is the union of all $\tdel{r}{i}$, $i\geq 2$.  We have that $\tcovonly_{r,k}=\emptyset$ for $k$ greater than the number of sites in $A$. Thus, in this case, $\del{r}{k}$ is empty for $k\geq 5$.
}
\label{fig:excisionidea}
\end{figure}

\subsection{Weak equivalence of $\covonly$ and $\tcovonly$}
We next prove that $\covonly$ and $\tcovonly$ are weakly equivalent.  The proof will use the following version of the usual mapping telescope construction  \cite[Section 3.F]{hatcher2005algebraic}: given any sequence of continuous maps
\[C = (C_1\rightarrow C_{2} \rightarrow \cdots  \rightarrow C_n),\]
let $M_C$, the \emph{mapping telescope of $C$}, be the quotient of the disjoint union \[\bigsqcup_{i=1}^{n} C_i\times I\]
given by gluing each $C_{i}\times \{1\}$ to $C_{i+1}\times \{0\}$ via the map $C_{i}\to C_{i+1}$.  It is easy to check that we have a deformation retraction $M_C\to C_n$.


\begin{theorem}\label{thm:covbifiltration}
The multicover bifiltration $\covonly$ is weakly equivalent to $\tcovonly$.
\end{theorem}

\begin{proof}
Our proof strategy is similar to ones used for sparse filtrations~\cite{cavanna15geometric,sheehy21sparse}.
We will observe that $\tcovonly$ is isomorphic to the nerve of a good cover of a bifiltration~$X$ which is weakly equivalent to $\covonly$.  The result then follows from the Persistent Nerve~Theorem.  

Let $X_{r,k}$ be the mapping telescope of the sequence 
\[\cov{r}{n}\hookrightarrow \cov{r}{n-1}\hookrightarrow \cdots \hookrightarrow \cov{r}{k}.\]  Letting $r$ and $k$ vary, the spaces $X_{r,k}$ assemble into a $([0,\infty)\times \N^{\mathrm{op}})$-indexed bifiltration $X$, and the deformation retractions $X_{r,k}\to \cov{r}{k}$ assemble into an objectwise homotopy equivalence $X\to \covonly$.  
Letting \[q:\bigsqcup_{i=k}^{n} \cov{r}{i}\times I\to X_{r,k}\] denote the quotient map, we have a good cover 
\[\firstcovername_{r,k}:=\{q(U\times I)\mid U\in \bigsqcup_{i=k}^n \vorcov{r}{i}\}\]
of $X_{r,k}$, and the collection of all such covers as $r$ and $k$ varies assembles into a good cover $\firstcovername$ of $X$.  

To finish the proof, it suffices to observe that $\nerve(\firstcovername)$ is isomorphic to $\tcovonly$.  First, note that vertices of $\nerve(\firstcovername_{r,k})$ correspond bijectively to non-empty elements of \[\bigsqcup_{i=k}^{n} \vorcov{r}{i},\] as do vertices of $\tcov{r}{k}$.  This gives us a bijection from the vertex set of  $\nerve(\firstcovername_{r,k})$ to the vertex set of $\tcov{r}{k}$.  In fact, this bijection is a simplicial isomorphism $\varphi_{r,k}:\nerve(\firstcovername_{r,k})\to \tcov{r}{k}$, because for all  \[\{U_1,\ldots,U_m\} \subset \bigsqcup_{i=k}^{n} \vorcov{r}{i},\]
we have \[\bigcap_{i=1}^m q(U_i\times I )\ne \emptyset\] if and only if both of the following conditions hold:
\begin{enumerate}
\item  there exists $j\geq k$ such that $\{U_1,\ldots,U_m\} \subset  \vorcov{r}{j} \cup \vorcov{r}{j+1},$
\item $\bigcap_{i=1}^m U_i\ne \emptyset$.
\end{enumerate}
The isomorphisms $\varphi_{r,k}$ are natural in $r$ and $k$, so assemble into an isomorphism 
from $\nerve(\firstcovername)$ to $\tcovonly$. 
\end{proof}

\subsection{Truncations of S-Del}
When the point cloud is large, it may be computationally difficult to construct the full bifiltration $\tcovonly$.  We instead consider, for $\kmax\in\N$, the bifiltration $\tcovonly^{\leq\kmax}$ constructed in the same way, but  disregarding all order-$k$ Voronoi cells with $k>\kmax$, \[\tcovonly^{\leq\kmax}_{r,k}:= \bigcup_{(\kmax-1) \geq i\geq k} \tdel{r}{i}.\] Note that $\tcovonly^{\leq |A|}=\tcovonly$.  
Viewing $\tcovonly^{\leq \kmax}$ as a $([0,\infty)\times \{0,\ldots,\kmax\}^{\mathrm{op}})$-indexed bifiltration, the proof of Theorem~\ref{thm:covbifiltration} adapts immediately to show that $\tcovonly^{\leq \kmax}$ is weakly equivalent to the restriction of $\covonly$ to $[0,\infty)\times \{0,\ldots,\kmax\}^{\mathrm{op}}$.

\subsection{Size of S-Del}
By the \emph{size} of a bifiltration $X$, we mean the number of simplices in the largest simplicial complex in $X$. 
If the sites $A\subset \R^d$ are not in general position, the size of $\tcovonly^{\leq\kmax}$ can be huge; indeed, if all points of $A$ lie on a circle in $\R^2$ and $r$ is at least the radius of this circle, then $\del{r}{k}$ has $2^{n \choose k}$ simplices, so $\tcovonly^{\leq \kmax}$ is at least as large.  However, if $A$ is in general position, then the situation is far better:

\begin{proposition}\label{Prop:Simplicial_Size_Bounds}
Let $A\subset\R^d$ be a set of $n$ sites in general position, with a constant dimension $d$. Then $\tcovonly^{\leq \kmax}$ has size 
$O(n^{ \floor*{\frac{d+1}{2}}} K^{\ceil*{\frac{d+1}{2}}})$.
In particular, $\tcovonly$ has size $O(n^{d+1})$.
\end{proposition}

In contrast to Proposition~\ref{Prop:Simplicial_Size_Bounds}, the number of vertices in Sheehy's subdivision-\v Cech model~\cite{Sheehy2012}  
grows exponentially with $n$, regardless of whether $A$ is in general position.
 
In brief, the idea of the proof is to bound the 
number of maximal simplices in $\tcovonly^{\leq \kmax}$ 
using a bound of Clarkson and Shor on the number of Voronoi vertices at levels $\leq k$~\cite{clarkson1989applications}.
The result then follows by observing that the dimension of the complex is a constant that only
depends on $d$. However, this dependence
on $d$ is doubly exponential, so the $O$-notation hides a large factor
if $d$ is large.  We now give details of the proof.

For $k\in \N$, and a bifiltration $X$, we let $X_{\infty,k}$ denote the largest simplicial complex of the form $X_{r,k}$, provided this exists.

In the following lemmas, $A\subset\R^d$ is assumed to be in general position.

\begin{lemma}\label{Lem:Dim_Bound_Fixed_k}
For all $k\in \N$, we have 
\[
\dim(\del{\infty}{k})
\leq {d+1 \choose \floor*{\frac{d+1}{2}}}.
\]
\end{lemma}
\begin{proof}
A simplex $\sigma\in \del{\infty}{k}$ is a set of order-$k$ Voronoi regions with a point of common intersection. Let $x$ be such a point.  Each order-$k$ Voronoi region $R\in \sigma$ is indexed by a subset of $A$ with $k$ elements, which we denote as $\mathrm{sites}(R)$.  Let 
\begin{align*}
s&=\min\, \{r\geq 0\mid x\in \cov{r}{k}\},\\
A_{on}&=\{p\in A\mid \|p-x\|=s\}, \\
A_{in}&=\{p\in A\mid \|p-x\|<s\}.
\end{align*}
By minimality of $s$, we have $|A_{in}|<k$, and by the assumption of general position, we have $|A_{on}|\leq d+1$.  Moreover, for all $R\in \sigma$, we have $A_{in}\subset \mathrm{sites}(R)$ and $\mathrm{sites}(R)\subset A_{in}\cup A_{on}$.  
Thus \[\dim(\sigma)\leq{|A_{on}| \choose k-|A_{in}|}-1\leq  {d+1\choose k-|A_{in}|}\leq {d+1 \choose \floor*{\frac{d+1}{2}}}.
\]
\end{proof}

\begin{lemma}\label{Lem:Dim_Bound_S-Del}
For all $\kmax\in \N$, we have 
\[\dim(\tcovonly^{\leq\kmax})\leq 2{d+1 \choose \floor*{\frac{d+1}{2}}}.\]
\end{lemma}

\begin{proof}
We first observe that
\[\dim(\tdel{\infty}{k})\leq \dim(\del{\infty}{k})+\dim(\del{\infty}{k+1})\leq 2 {d+1 \choose \floor*{\frac{d+1}{2}}},\]
by noting that the dimension on the left is upper bounded by the sum of the~maximal 
number of order-$k$ Voronoi regions and order-$(k+1)$ Voronoi regions
that meet at a fixed point in $\R^d$. Since these two summands equal
$\dim(\del{\infty}{k})-1$ and $\dim(\del{\infty}{k+1})-1$, we have the first
inequality, and the second one follows from the previous Lemma~\ref{Lem:Dim_Bound_Fixed_k}.  Since each simplicial complex in $\tcovonly^{\leq\kmax}$ is a union of simplicial complexes $\tdel{r}{k}$, it follows that 
\[\dim(\tcovonly^{\leq\kmax})\leq 2{d+1 \choose \floor*{\frac{d+1}{2}}},\]
as desired.
\end{proof}

\begin{lemma}\label{lem:maximal_simplices}
If $n>d$, then for any $k\in \N$, the number of maximal simplices in $\tdel{\infty}{k}$ is at most $V_k+V_{k+1}$, where $V_k$, $V_{k+1}$ denote
the number of Voronoi vertices at level $k$ or $k+1$, respectively.
\end{lemma}
\begin{proof}
Recall that $\vor{k}$ and $\vor{k+1}$ denote the
order-$k$ and order-$(k+1)$ Voronoi diagrams, respectively.  Let $\vor{k+1/2}$ denote their \emph{overlay}, i.e., the polyhedral subdivision of $\R^n$ induced by the closed polyhedra of the form $R\cap S$, where $R$ and $S$ are Voronoi regions of order $k$ such that $\dim(R\cap S)=d$.
The subdivision $\vor{k+1/2}$ is called the \emph{Voronoi diagram of degree-$(k+1)$}~\cite{Edelsbrunner2018}.  It has been studied in~\cite{Edelsbrunner1986}. 

Maximal simplices of $\tdel{\infty}{k}$ correspond bijectively to cells in $\vor{k+1/2}$ with no proper faces.  We claim that the only such cells are vertices.  If $k\in \{0,n-1,n\}$, then either $\vor{k}$ or $\vor{k+1}$ contains only the region $\R^d$, and it is clear that the claim holds.  For $0< k<n-1$, we observe that since $A$ is in general position and $n>d$, each cell of either $\vor{k}$ or $\vor{k+1}$ is bounded by a vertex, which implies that each cell of $\vor{k+1/2}$ is bounded by a vertex.  The claim follows.

To finish the proof, we will show that the vertex set of $\vor{k+1/2}$ 
is the union of the vertex sets of $\vor{k}$ and $\vor{k+1}$.  
It suffices to show that for any cell $\sigma\subset \R^d$ of $\vor{k}$ with $\dim(\sigma)<d$, $\sigma$ is contained in a cell of $\vor{k+1}$; then $\vor{k+1/2}$ does
not subdivide the $(d-1)$-skeleton $\vor{k}$, implying that no new vertex is created.  

Let $x\in \sigma$.  As in Lemma \ref{Lem:Dim_Bound_Fixed_k}, we let 
\begin{align*}
s&=\min\, \{r\geq 0\mid x\in \cov{r}{k}\},\\
A_{on}&=\{p\in A\mid \|p-x\|=s\}, \\
A_{in}&=\{p\in A\mid \|p-x\|<s\},
\end{align*}
and we note that $|A_{in}|<k$.  Moreover, since $\dim(\sigma)<d$, we have $|A_{on}|\geq 2$.  
The set of order-$k$ Voronoi regions containing $x$ (and hence any $y\in \sigma$) is 
\[\{\vorreg{A_{in}\cup \tilde A}\mid \tilde A\subset A_{on},\ |\tilde A|=k-i\}.\]
Therefore $A_{on}$ and $A_{in}$ are independent of our choice of $x\in \sigma$.  Hence, for any $x\in \sigma$, the set of order-$(k+1)$ Voronoi regions containing $x$ is 
\[\{\vorreg{A_{in}\cup \tilde A}\mid \tilde A\subset A_{on},\ |\tilde A|=k+1-i\}.\]
The intersection of these Voronoi regions is the closure of a cell of $\vor{k+1}$ which contains $\sigma$.  (In fact, similar reasoning shows that  
if $|A_{in}|+|A_{on}|\geq k+2$, then $\sigma$ is a cell in $\vor{k+1}$, though this is not needed for the argument.)

\end{proof}

\begin{proof}[Proof of Proposition~\ref{Prop:Simplicial_Size_Bounds}]
Assume that $n>d$.  Since \[\tcovonly^{\leq\kmax}=\bigcup_{k=0}^{\kmax-1}\tdel{\infty}{k},\]
Lemma~\ref{lem:maximal_simplices} implies that the number of its maximal simplices is at most
\[(V_0+V_1)+(V_1+V_2)+\ldots + (V_{\kmax-1}+V_{\kmax})\leq 2 V_{\leq\kmax},\]
where $V_{\leq\kmax}$ is the number of Voronoi vertices at levels $\leq k$. 
Since the number of simplices of a $d$-dimensional simplicial complex with $m$
maximal simplices is at most $2^{d+1}m$, Lemma~\ref{Lem:Dim_Bound_S-Del} implies that the size of $\tcovonly^{\leq \kmax}$ is bounded above by 
\[
2^{\left(2{d+1 \choose \floor*{\frac{d+1}{2}}}+1\right)} \cdot V_{\leq\kmax}.
\]
Since the dimension bound only depends on the constant~$d$, it is therefore
enough to bound $V_{\leq\kmax}$.

By a result of Clarkson and Shor~\cite{clarkson1989applications},
we have that $e_{\leq\kmax}$, the number of Voronoi regions (i.e., $d$-dimensional Voronoi cells) at levels $\leq\kmax$ is
\[e_{\leq\kmax}=O(|A|^{ \floor*{\frac{d+1}{2}}} \kmax^{\ceil*{\frac{d+1}{2}}}).\]
This bound hides a constant that depends doubly exponentially in $d$.  
Moreover, by a counting argument appearing in~\cite[Theorem 3.3]{edelsbrunner1987algorithms}, we have that the total size of the Voronoi diagram
at level $k$ is bounded by
\[
O(\mathrm{max}\, \{e_i\mid k-d+1\leq i\leq k+d-2\})
\]
with $e_i$ the number of Voronoi regions at level $i$.
The same bound applies to $V_{k}$ as well. A simple calculation shows
that $V_{\leq \kmax}$ is then bounded by $(2d-2)e_{\leq \kmax}=O(e_{\leq\kmax})$.
\end{proof}



\section{The rhomboid bifiltration}
\label{sec:rhomboid}

\subsection{The rhomboid tiling}
Let  $\Xp \subseteq \R^d$ be a set of $n$ sites in general position, and
$S$ an arbitrary $(d-1)$-sphere in $\R^d$.
Then $S$ yields a decomposition $\Xp=\Xin \sqcup \Xon \sqcup \Xout$ 
with $\Xin$ the sites in the interior of $S$, $\Xon$ the sites on the sphere,
and $\Xout$ the sites in the exterior of $S$. We define the \emph{combinatorial rhomboid} of $S$ to be the collection of sets 
\begin{equation}
\vx{S} := \left\{\Xin \cup \qv \mid \qv \subseteq \Xon\right\}.
\label{eqn:newvx}
\end{equation}
We call elements of $\vx{S}$ \emph{combinatorial vertices}, and call
\begin{equation}
\rhomonly(\Xp) = \left\{\vx{S} \mid \text{$S$ is a sphere in $\R^d$}\right\}
\end{equation}
the \emph{(combinatorial) rhomboid tiling} of $\Xp$.  Elements of $\rhomonly(\Xp)$ are called \emph{rhomboids}.  Since $\Xp$ is fixed throughout,
we write $\rhomonly$ instead of $\rhomonly(\Xp)$.

As observed in \cite{Edelsbrunner2018}, the combinatorial rhomboid tiling can be geometrically
realized as a polyhedral cell complex~\cite[Def~2.38]{kozlov-book}. For that, a combinatorial vertex $\{a_1,\ldots,a_k\}$
(where $a_1,\ldots,a_k$ are sites in $\R^d$) is embedded as \linebreak $(\sum_{i=1}^k a_i,-k)$ in $\R^{d+1}$.
We call $k$ the \emph{depth} of the vertex. Embedding a \linebreak combinatorial rhomboid as the convex hull
of its embedded vertices yields an actual rhomboid in $\R^{d+1}$ whose dimension equals the cardinality of $\Xon$
in the corresponding partition of $\Xp$. The collection of these rhomboids is the \emph{(geometric) rhomboid tiling}
for $\Xp$.
We illustrate the construction in Figure~\ref{fig:rhomboids}. In what follows, we identify
vertices and rhomboids with their combinatorial description. In particular, we will
use $\rhomonly$ both for the combinatorial and the geometric rhomboid tiling.

\begin{figure}
        \centering   
        \includegraphics[width=0.7\textwidth,page=2]{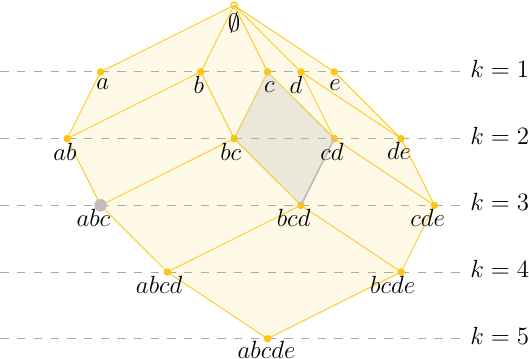}
        \caption{The rhomboid tiling of $5$ points on the real line.
            The highlighted 2-rhomboid $\rho$ defined by $\xin{\rho} = \{c\}$ and $\xon{\rho} = \{b,d\}$
            is the convex hull of the points $c$, $bc$, $cd$,
            and $bcd$, simplifying the labels here and, e.g., writing $bcd$ instead of the cell complex associated to $\{b,c,d\}$.
            The horizontal line at depth $k$ intersects the tiling in
            the order-$k$ Delaunay mosaic.}
        \label{fig:rhomboids}
\end{figure}

\begin{proposition}[{\cite[Proposition 4.8]{osang2021multi},\cite[Section 1.2]{edelsbrunner1987algorithms}}]\label{Prop:RT_Size}
The number of cells (of all dimensions) in $\rhomonly$ is at most $\frac{2^{d+1}}{(d+1)!}(n+1)^{d+1}\leq 2(n+1)^{d+1}$.
\end{proposition}

For any rhomboid $\rho\in\rhomonly$, we let 
\[r_\rho=\inf\, \{r \mid \textup{there exists a sphere $S$ of radius $r$ such that $\vx{S}=\rho$}\}.\] It is easily checked that if $\rho'$ is a subset of $\rho$, then $r_{\rho'}\leq r_{\rho}$,
and therefore, for any $r\geq 0$, the sublevel set
$\rhomonly_{r}=\{\rho\in\rhomonly\mid r_{\rho}\leq r\}$
is also a polyhedral complex.

\subsection{Slicing}
Next, we slice the rhomboid tiling by cutting every rhomboid along the \linebreak hyperplanes $\{x \in \R^{d+1} \mid -x_{d+1}=k\}$
with $k=0,\ldots,n$.
In this way, a rhomboid decomposes into its intersections with these hyperplanes and with slabs of the form $\{x \in \R^{d+1} \mid k\leq -x_{d+1}\leq k+1\}$.
The resulting polyhedra again form a polyhedral complex 
that we call the \emph{sliced rhomboid tiling} $\trhomonly$.
We refer to its cells as \emph{sliced rhomboids}.  

For a sliced rhomboid $\rho$, we define $k_\rho$ as the minimum depth among its vertices. Moreover, there is a unique (unsliced) rhomboid $\rho'$ of smallest dimension
that contains $\rho$, and we define $r_\rho:=r_{\rho'}$. Define
\[
\trhom{r}{k}:=\{\rho\in\trhomonly\mid r_{\rho}\leq r, k_\rho\geq k\}
\]
and observe that for $r\leq r'$ and $k\geq k'$, we have $\trhom{r}{k}\subseteq\trhom{r'}{k'}$. Hence, $(\trhom{r}{k})_{(r,k)\in [0,\infty)\times\N^{\mathrm{op}}}$
is a bifiltration of combinatorial cell \linebreak complexes. Again, we will abuse notation and use the symbol $\trhomonly$ both for the sliced rhomboid tiling
and the bifiltration $(\trhom{r}{k})_{(r,k)\in [0,\infty)\times\N^{\mathrm{op}}}$.   As shown in~\cite{Edelsbrunner2018}, the restriction of   $\trhomonly$ to cells in the hyperplane $-x_{d+1}=k$ is the order-$k$ Delaunay mosaic, i.e., the geometric dual of the order-$k$ Voronoi diagram.

\subsection{Comparison of S-Rhomb and S-Del} 

The next lemma establishes a close relationship between the bifiltrations $\trhomonly$ and $\tcovonly$.  It is closely related to \cite[Theorem 1 and Lemma 2]{Edelsbrunner2018}.

\begin{lemma}\mbox{}
\label{lem:correspondence}
For all $(r,k)\in [0,\infty)\times\N$,
\begin{enumerate}
\item The vertex sets of $\trhom{r}{k}$ and $\tcov{r}{k}$ are equal.
\item The vertices of each sliced rhomboid in $\trhom{r}{k}$ span a simplex in $\tcov{r}{k}$.
\item The vertices of each simplex in $\tcov{r}{k}$ are contained in a sliced rhomboid of $\trhom{r}{k}$.
\end{enumerate}
\end{lemma}

\begin{proof}
For the first part, note that a set of sites $v=\{a_1,\ldots,a_{k'}\}$ is a vertex in $\trhom{r}{k}$ if and only if $k'\geq k$ and for all $r'>r$ there is a sphere $S$ of radius at most $r'$ whose associated decomposition $\Xp=\Xin \sqcup \Xon \sqcup \Xout$ satisfies 
\begin{equation}\label{Eq:Vertex}
\Xin\subseteq v\subseteq\Xin\cup\Xon.
\end{equation}  Note that \eqref{Eq:Vertex} holds if and only if the center of $S$ has $v$ among its $k'$ closest sites, which is equivalent to the condition that the order-$k'$
Voronoi region $\vorreg{v}$ contains the center of $S$.  Thus, such a sphere $S$ exists if and only $v$ is a vertex of $\del{r'}{k'}$.  
Thus, $v\in \trhom{r}{k}$ if and only if $k'\geq k$ and $v\in \del{r'}{k'}$ for all $r'>r$.  But the latter holds if and only if $v\in \tcov{r}{k}$ because, first, the vertex sets of $\tcov{r'}{k}$ and $\sqcup_{k'\geq k} \del{r'}{k'}$ are equal and, second,  $v\in \tcov{r'}{k}$ for all $r'>r$ implies $v\in \tcov{r}{k}$.  Thus, the vertex sets of $\trhom{r}{k}$ and $\tcov{r}{k}$ are equal, as claimed.

For the second part, let $v_1,\ldots,v_m$ denote the vertices of a sliced rhomboid in $\trhom{r}{k}$.
Let $\rho$ denote the smallest rhomboid of $\rhomonly$ containing this sliced rhomboid.  For any $r'>r$, there exists a sphere $S$ of radius at most $r'$ with $\rho_S=\rho$. 
Let $x$ denote the center of $S$ and $\Xp=\Xin \sqcup \Xon \sqcup \Xout$ be the decomposition with respect to $S$.
Now, each $v_i$ is the union of $\Xin$ with a subset of $\Xon$, and hence, the point
$x$ belongs to the Voronoi region of $v_i$. Since $i$ is arbitrary, it follows that all Voronoi
regions intersect, and $x$ has distance $\leq r'$ to each site in each $v_i$, so $v_1,\ldots,v_m$ span a simplex in $\tcov{r'}{k}$.  Since this holds for all $r'>r$, this simplex is also contained in $\tcov{r}{k}$.

For the third part, consider vertices $v_1,\ldots,v_m$ that span a simplex in $\tcov{r}{k}$.
Assume that some $v_i$ has order $k'\geq k$ and that the remaining vertices are of order $k'$ or $k'-1$.
Since $v_1,\ldots,v_m$ span a simplex, the corresponding higher-order Voronoi regions, intersected with balls of radius $r$ around the involved sites, have non-empty intersection.  Let $x$ be a point in this intersection.  

If $k'=0$, then $m=1$ and $v_1=\emptyset$.  As $\{\emptyset\}$ is a combinatorial rhomboid of dimension 0, the desired result holds in this case.  If $k'=1$ and $x$ is a site, then either $m=1$ and $v_1=\{x\}$, or else $m=2$, in which case $\{v_1,v_2\}=\{\{x\},\emptyset\}$.  In either case, the desired result again holds.  

Otherwise, there is smallest sphere $S$ centered at $x$ that includes $k'$ sites (either on the sphere or in its interior). $S$ induces a partition $\Xp=\Xin \sqcup \Xon \sqcup \Xout$ and a rhomboid $\rho$.
Each vertex $v_i$ of order $k'$ must contain all sites of $\Xin$, and some subset of the sites of $\Xon$, meaning that $v_i$ is in $\rho$.
Furthermore, at least one site lies on $S$; otherwise, there would be a smaller sphere.
This implies that each vertex $v_j$ of order $k'-1$ also has to contain all sites of $\Xin$ and some subset of $\Xon$.
Thus, all vertices lie in $\rho$, and in particular in its $(k'-1,k')$-slice.
Finally, observe that because one of the sites lies on $S$, the radius of $S$ is the distance
of that site to $x$, which is at most $r$.
\end{proof}

\begin{remark}\label{Rem:Injection}
Parts 1 and 2 of Lemma~\ref{lem:correspondence} establish that we have a vertex-preserving injection $\mathcal J$ from the cells of $\trhom{r}{k}$ to the simplices of $\tcov{r}{k}$.  Moreover, the third part of Lemma~\ref{lem:correspondence} implies that $\mathcal J$ restricts to a bijection from the maximal cells of $\trhom{r}{k}$ to the maximal simplices of $\tcov{r}{k}$.  Figure~\ref{fig:diff} illustrates that $\mathcal J$ itself needn't be a bijection.

We note that $\mathcal J$ does not preserve dimension:  For $\nu$ a cell of $\trhomonly$ spanned by vertices of cardinality $k$  (i.e., a cell in the order-$k$ Delaunay mosaic), we have $\dim(\nu)\leq d$.  But if $d\geq 3$, it can be that $\dim(\mathcal J(\nu))>d$, even when the sites are in general position. For a cell $\nu$ of $\trhom{r}{k}$ spanned by vertices of cardinality $k$ and $k+1$, we may have $\dim(\nu)\ne \dim(\mathcal J(\nu))$ even for $d=2$, see Figure~\ref{fig:diff}.
\end{remark}
\begin{figure}
        \centering  
        \includegraphics[width=0.8\textwidth]{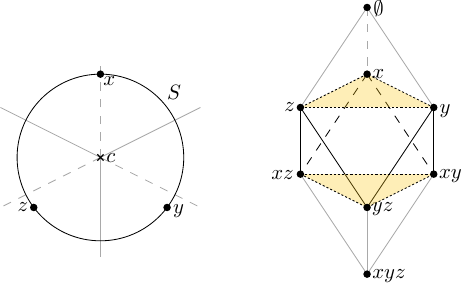}
        \caption{An illustration of the difference between $\tcovonly$ and $\trhomonly$ for three points in the plane.
        The order-$1$ Voronoi regions of the points $\{x\}$, $\{y\}$, and $\{z\}$ intersect in $c$, as do the order-$2$ Voronoi regions of $\{x,y\}$, $\{y,z\}$, and $\{x,z\}$. Consequently, $\tcovonly$ contains a $5$-simplex, but the corresponding cell in $\trhomonly$ on the same vertex set is a $3$-dimensional triangular skew prism.  Many simplices in $\tcovonly$ do not correspond to any cell in $\trhomonly$, e.g., the 1-simplex $\{x,yz\}$.}
        \label{fig:diff}
\end{figure}

%

\begin{theorem}
\label{thm:trhom_tcov}
The bifiltrations $\trhomonly$ and $\tcovonly$ are weakly equivalent.
\end{theorem}
\begin{proof}
We define good covers 
of both bifiltrations: First, for $\trhom{r}{k}$, we choose the cover that
consists of all its cells. This is a good cover because the cells are convex. The collection
of these covers over all choices of $r$ and $k$ yields a good cover $\mathcal{U}$ of the bifiltration $\trhomonly$, and the Persistent Nerve Theorem then gives objectwise homotopy equivalences
\[
\begin{tikzcd}
 \trhomonly  &  \arrow[labels=above]{l}{\simeq} \arrow{r}{\simeq} \deltaspace{\mathcal{U}}{\trhomonly}  & \nerve(\mathcal{U}).
\end{tikzcd}
\]
for some intermediate bifiltration $\deltaspace{\mathcal{U}}{\trhomonly}$.

Moreover, we obtain a cover $\mathcal{V}_{r,k}$ of $\tcov{r}{k}$ whose elements are the simp\-lices
spanned by the vertices of the sliced rhomboids in $\trhom{r}{k}$.
By the second part of Lemma~\ref{lem:correspondence},
these cover elements indeed exist. Moreover, in view of Remark~\ref{Rem:Injection}, every maximal simplex of  $\tcov{r}{k}$ is an element of $\mathcal{V}_{r,k}$.  We thus obtain a good cover $\mathcal{V}$ of $\tcovonly$.  Applying the Persistent Nerve Theorem again, we obtain objectwise homotopy equivalences
\[
\begin{tikzcd}
 \tcovonly  &  \arrow[labels=above]{l}{\simeq} \arrow{r}{\simeq} \deltaspace{\mathcal{V}}{\tcovonly}  & \nerve(\mathcal{V}).
\end{tikzcd}
\]

Finally, $\nerve(\mathcal{U})$ and $\nerve(\mathcal{V})$ are isomorphic: The elements of $\mathcal{U}$ and of~$\mathcal{V}$ are in 1-to-1 correspondence, with corresponding cover elements having the same vertex set. In both cases, an intersection of cover elements is non-empty if and only if the elements share a vertex, which is determined
by their vertex sets. Hence, we have objectwise homotopy equivalences
\[
\begin{tikzcd}
 \trhomonly  &  \arrow[labels=above]{l}{\simeq} \arrow{r}{\simeq} \deltaspace{\mathcal{U}}{\trhomonly}  & \nerve(\mathcal{U}) \arrow["\cong"]{r} & \nerve(\mathcal{V}) &  \arrow[labels=above]{l}{\simeq} \arrow{r}{\simeq} \deltaspace{\mathcal{V}}{\tcovonly}  & \tcovonly.
\end{tikzcd}
\]
\end{proof}

\subsection{Unslicing the rhomboid}
Next, we define a bifiltration on the (unsliced) rhomboid tiling. Recall that for a rhomboid $\rho$, 
we already defined $r_\rho$ as the radius of the smallest sphere that gives rise to that rhomboid.
As in the sliced version, we define $k_\rho$ as the minimal depth among the vertices of $\rho$,
and \[\Rhomboid{r}{k}:=\{\rho\mid r_\rho\leq r, k_\rho\geq k\}.\]
This yields a bifiltration $(\Rhomboid{r}{k})_{(r,k)\in [0,\infty)\times\N^{\mathrm{op}}}$, 
which we denote by $\rhomonly$.

\begin{lemma}
The bifiltrations $\rhomonly$ and $\trhomonly$ are weakly equivalent.
\end{lemma}
\begin{proof}
For a rhomboid $\rho$ in $\rhomonly$,
set $\altkmin$ as the minimum depth and $\altkmax$ as the maximum depth among the vertices in $\rho$.  Note that $k_\rho=\altkmin$.
For $r$ and $k'$ fixed, we say $\rho$ is \emph{dangling}
if $r_\rho\leq r$ and $\altkmin < k' < \altkmax$.  If $\rho$ is dangling then $\rho\not\in \Rhomboid{r}{k}$, but some of the slices of $\rho$ 
are contained in $\trhom{r}{k}$. In fact, all cells of $\trhom{r}{k}$ not contained in (the geometric realization of) $\Rhomboid{r}{k}$ are of this form. 
For example, taking $k=2$ and $r$ very large, the shaded rhomboid $\{c,bc,cd,bcd\}$ of Figure \ref{fig:rhomboids} is dangling.  $\trhom{r}{2}$ contains the cell $\{bc,cd,bcd\}$ but $\Rhomboid{r}{2}$ does not.

Observe that there is a deformation retraction of $\trhom{r}{k}$ onto $\Rhomboid{r}{k}$
which, for each dangling rhomboid $\rho$, ``pushes" \[\rho\cap \{x\in \R^{d+1}\mid x_{d+1}\geq k\}\] onto the boundary of $\rho$;  
for instance, in the example above $\{bc,cd,bcd\}$ is pushed onto $\{bc,bcd\}\cup \{cd,bcd\}$.
Thus, for every choice of $r$ and $k$, the inclusion $\Rhomboid{r}{k}\hookrightarrow \trhom{r}{k}$ is a homotopy equivalence.
Moreover, these inclusions commute with the inclusion maps in $\rhomonly$ and $\trhomonly$, hence define an objectwise homotopy equivalence.
\end{proof}

Combining the previous lemma with Theorem~\ref{thm:trhom_tcov} and Theorem~\ref{thm:covbifiltration} yields the following~result:
\begin{theorem}\label{thm:rhom_main}
The bifiltrations $\rhomonly$ and $\covonly$ are weakly equivalent.
\end{theorem}

\begin{remark}[Size of the Rhomboid Bifiltration]\label{Rem:Size_Bound_Rhom}
In view of Proposition~\ref{Prop:RT_Size}, $\rhomonly$ has at most $2(n+1)^{d+1}=O(n^{d+1})$ cells.  
One can also bound the size of a truncated version of $\rhomonly$, defined analogously to the truncation of $\tcovonly$ considered in Proposition~\ref{Prop:Simplicial_Size_Bounds}.  Indeed, $\rhomonly$ is clearly smaller (in terms of number of cells) than $\trhomonly$, and by Remark~\ref{Rem:Injection}, $\trhomonly$ is at least as small as $\tcovonly$.  Moreover, this extends to truncations of these bifiltrations.  Thus, the size bound of Proposition~\ref{Prop:Simplicial_Size_Bounds} also holds for truncations of $\rhomonly$.
\end{remark}

%

\subsection{Computation}
In~\cite{EdOs20,osang2021multi}, a relatively simple algorithm is given for computing the rhomboid bifiltration.  (In fact,  \cite{EdOs20} explicitly considers only the computation of the rhomboid tiling and the radius $r_{\rho}$ of each rhomboid $\rho$; the rhomboid bifiltration is not mentioned.  But it is trivial to extend the algorithm to compute the depth $k_{\rho}$ of each rhomboid $\rho$, thus computing the rhomboid bifiltration.)

 We briefly outline the approach.  Given a $(d+1)$-dimensional rhomboid $\rho$, let $\sigma$ denote the intersection of $\rho$ with the hyperplane $k=k_\rho+1$.  We call $\sigma$ the \emph{generation-1} slice of $\rho$.  Note that $\sigma$ is a $d$-simplex. 
Given the combinatorial vertices of $\sigma$, we can easily recover $\rho$ \cite[Lemma 2]{EdOs20}. 

The algorithm of \cite{EdOs20} computes the rhomboid tiling by computing the \linebreak generation-1 slices of all $(d+1)$-dimensional rhomboids.  For each $k$ in increasing order, a weighted Delaunay triangulation $W_k$ is computed which triangulates the order-$k$ Delaunay mosaic and has the same vertex set.  Given the vertex set, $W_k$ can be computed via any algorithm for weighted Delaunay triangulation computation, e.g., via a $d+1$-dimensional convex hull computation \cite[Section 4.4.4]{boissonnat2018geometric}.  We explain below how the vertex set is computed.   

A simple combinatorial criterion \cite[Lemma 3]{EdOs20} tells us whether a $d$-simplex in $W_k$ is a generation-1 slice of a rhomboid.  Thus, one can efficiently identify \linebreak all generation-1 slices in the triangulation by iterating through all the $d$-simplices of $W_k$.  In this way, we identify all $(d+1)$-dimensional rhomboids $\rho$ with \linebreak $k_\rho= k-1$.  

It remains to explain how the vertex set of $W_k$ is computed.  The vertices of $W_1$ are just the sites $A$.  For $k\geq 2$, \cite[Lemma 3]{EdOs20} establishes that every vertex $v$ in $W_k$ appears in a rhomboid $\rho$ with $k_\rho \leq k-2$.  We thus discover $\rho$, and hence $v$, by the time we finish processing $W_{k-1}$.  

\paragraph{Complexity}
The complexity of this algorithm is discussed in~\cite[Section 4]{EdOs20}.  While an explicit runtime bound is not given, it is easy to extract naive bounds from the discussion; we now do so.  We distinguish between two contributions to the runtime:
\begin{enumerate}
\item computing $W_k$ at all levels $k$, given the vertices,
\item checking, for each $k$, whether each $d$-simplex in $W_k$ is a generation-1 slice and if so, storing the corresponding rhomboid and its faces. 
\end{enumerate}
 The latter requires $O(k)$ time per $d$-simplex in $W_k$.  Hence, since the rhomboid tiling has size $O(n^{d+1})$ (see Remark \ref{Rem:Size_Bound_Rhom}), the total time required over all \linebreak $d$-simplices is $O(n^{d+2})$.  

The complexity of computing the triangulations $W_k$ depends on a choice of algorithm for computing weighted Delaunay triangulations.   Some well-known algorithms have output-sensitive complexity bounds.  For example, in the case $d=3$, a weighted Delaunay triangulation of $p$ points can be computed by the algorithm of \cite{chan1997primal} in time $O((p+m)\log^2 m)$, where $m$ is the size of the output.  In our setting, the total size of all the $W_k$ is $O(n^{d+1})$ because the size of each $W_k$ differs from the size of the order-$k$ Delaunay mosaic by at most a constant factor.  Hence for $d=3$, which is arguably the case of primarily interest, the total time to compute all of the triangulations $W_k$ is $O(n^{4}\log^2 n)$.  Therefore, the total cost of computing the rhomboid bifiltration is $O(n^{5}+n^{4}\log^2 n)=O(n^{5}$).  

For arbitrary $d$, the approach of \cite[Section 4.4.4]{boissonnat2018geometric} computes the weighted Delaunay triangulation of $p$ points in $\R^d$  in time $O(p\log p+ p^{\lceil \frac{d}{2}\rceil})$.  In our setting, each vertex of each $W_k$ is a vertex of the rhomboid tiling, so there are a total of $O(n^{d+1})$ vertices among all $W_k$.  Thus, for $d\geq 3$ the time required to compute all $W_k$ is $O(n^{(d+1)\lceil \frac{d}{2}\rceil})$, and the runtime of the full algorithm satisfies the same asymptotic bound.  This bound is rather large, but it seems likely that it could be improved via a more careful analysis.

\paragraph{Implementation}
The above algorithm has been implemented in the software package \textsc{rhomboidtiling}\footnote{\url{https://github.com/geoo89/rhomboidtiling}}~\cite{EdOs20}.  The code computes the sliced and unsliced bifiltrations~$\trhomonly$ and $\rhomonly$ as well as their free implicit representations (FIREPs), i.e., chain complexes~\cite{lesnick2019computing}. \textsc{rhomboidtiling} is written in C++, using the \textsc{Cgal} library\footnote{CGAL, Computational Geometry Algorithms Library, \url{https://www.cgal.org}} for geometric primitives. The current version accepts only 2- and 3-dimensional inputs, but all steps readily generalize to higher dimensions; adding support for higher-dimensional inputs is a matter of \linebreak software design rather than algorithm development.  That said, handing higher-dimension inputs of practical size is likely to be computationally expensive.

 

\section{Experiments}\label{sec:experiments}

We performed experiments on point sets in $\R^2$ and $\R^3$.  We provide a brief summary here; for detailed results, see Appendix~\ref{app:experiments}.  We sampled points \linebreak uniformly at random from $[0,1]^2$ and $[0,1]^3$, from a disk, from an annulus, and from an annulus with noise added.  We computed the rhomboid bifiltrations  $\rhomonly^{\leq\kmax}$ and  $\trhomonly^{\leq\kmax}$.  We then used \textsc{mpfree}\footnote{\url{https://bitbucket.org/mkerber/mpfree}} to compute minimal presentations of 2-parameter persistent homology of our bifiltrations.
 
In one set of experiments, we found that $\rhomonly^{\leq\kmax}$  is up to 47\% smaller than $\trhomonly^{\leq\kmax}$, and can be computed more than 20\% faster.  The experi\-ments suggest that the relative performance of $\rhomonly^{\leq\kmax}$ improves with \linebreak increasing $\kmax$.

\begin{figure}[h]
 	\centering
	\includegraphics[width=0.305\textwidth]{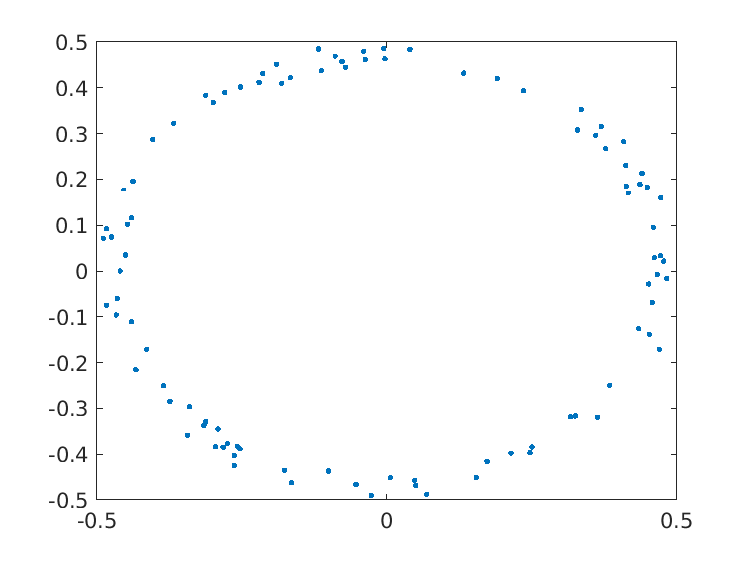} 
	\includegraphics[width=0.18\textwidth]{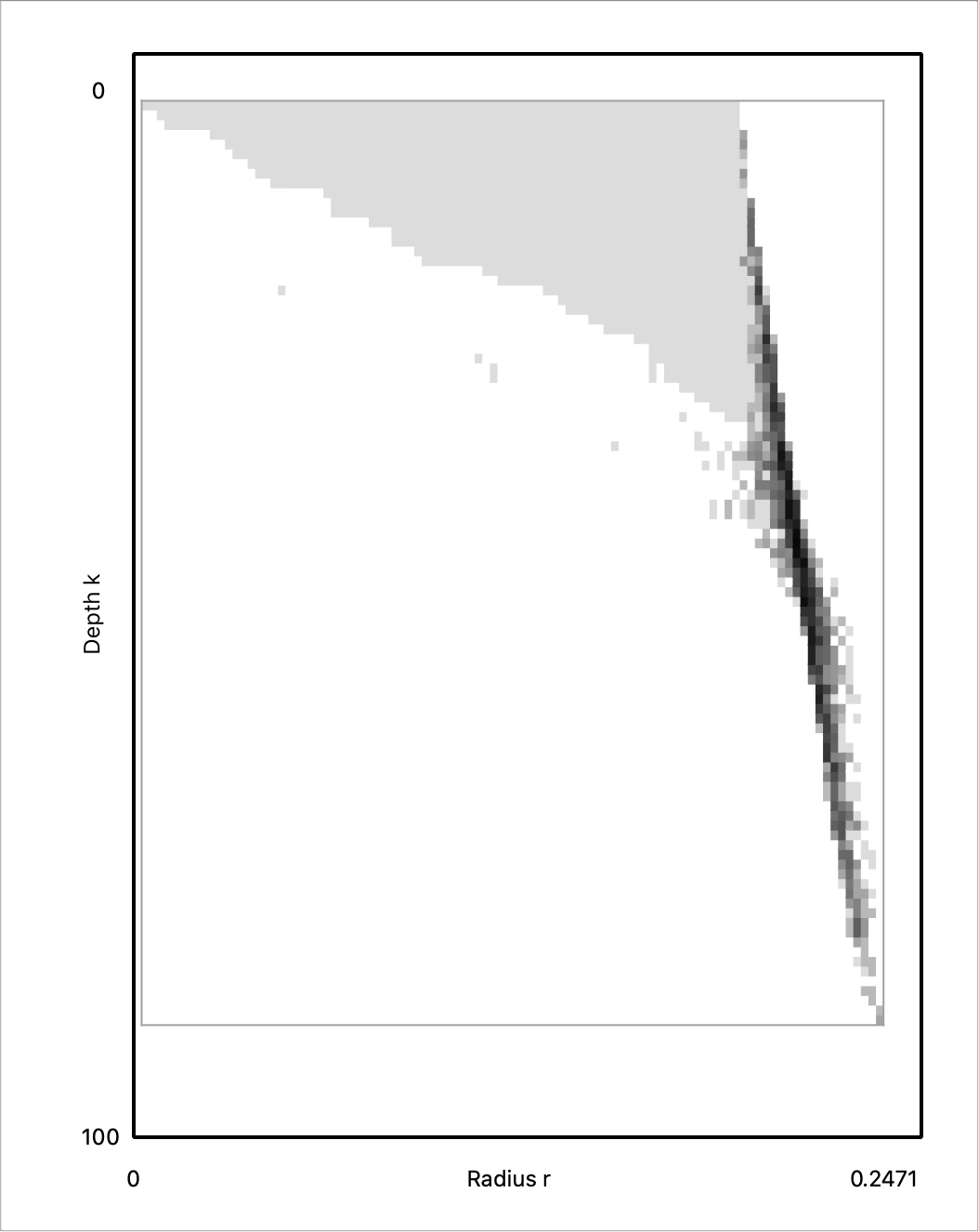} 
	\includegraphics[width=0.305\textwidth]{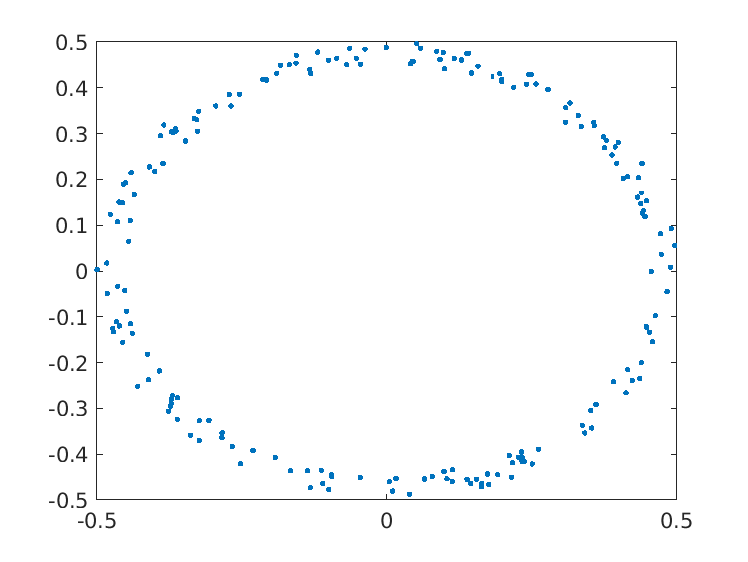} 
	\includegraphics[width=0.18\textwidth]{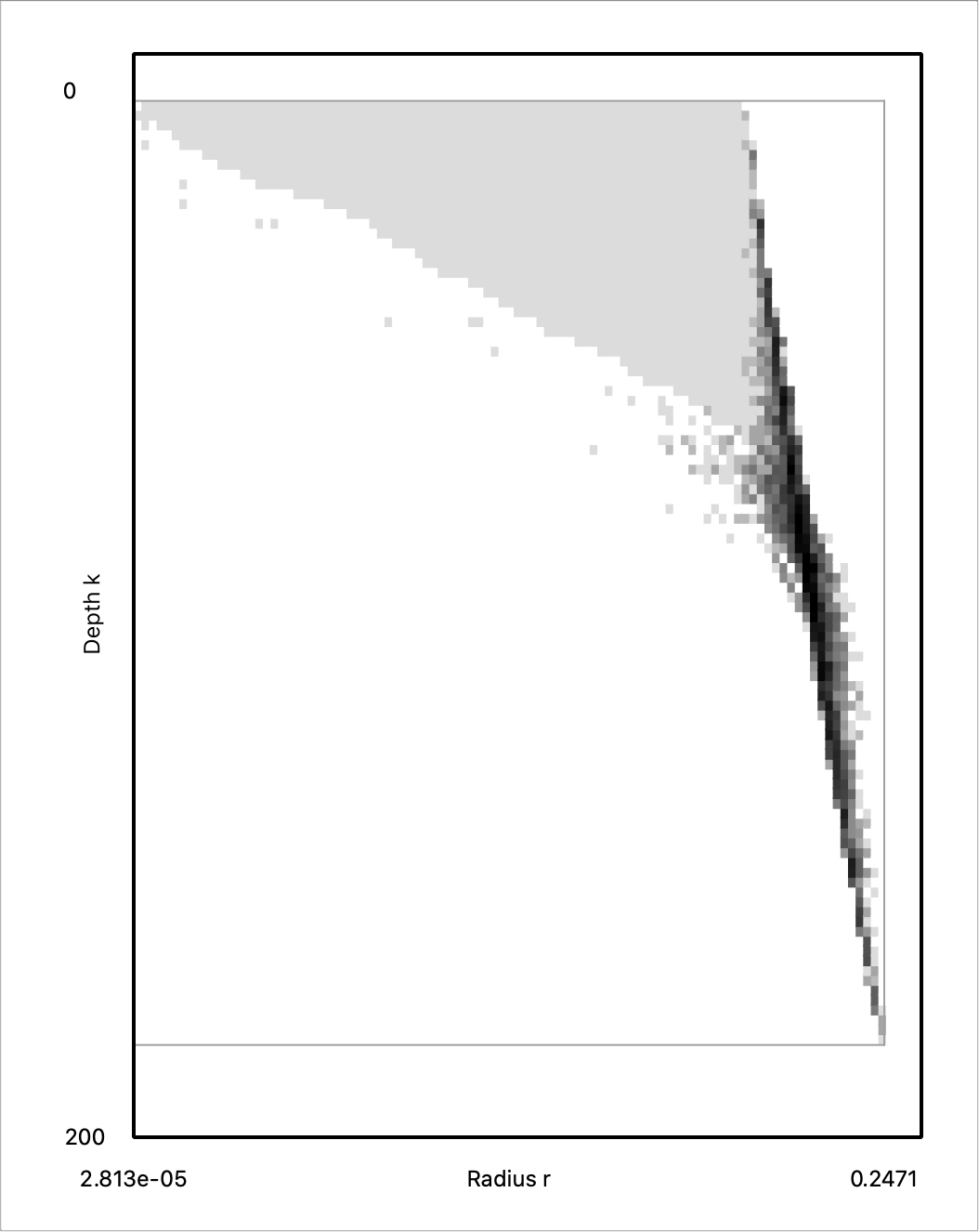}  
	\\ 
	\includegraphics[width=0.305\textwidth]{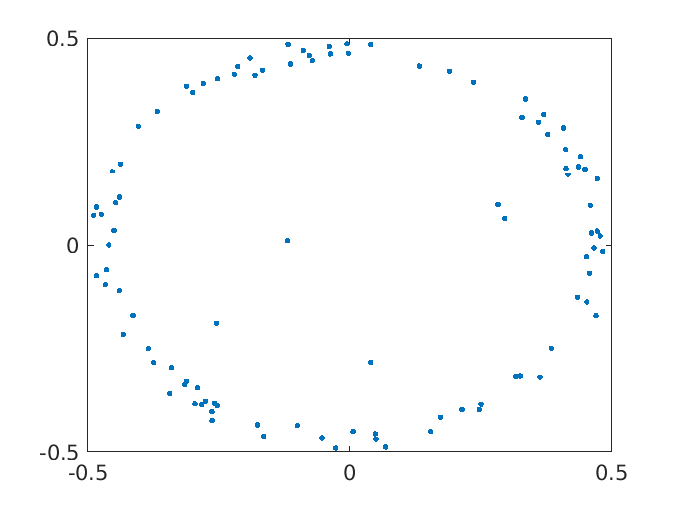} 
	\includegraphics[width=0.18\textwidth]{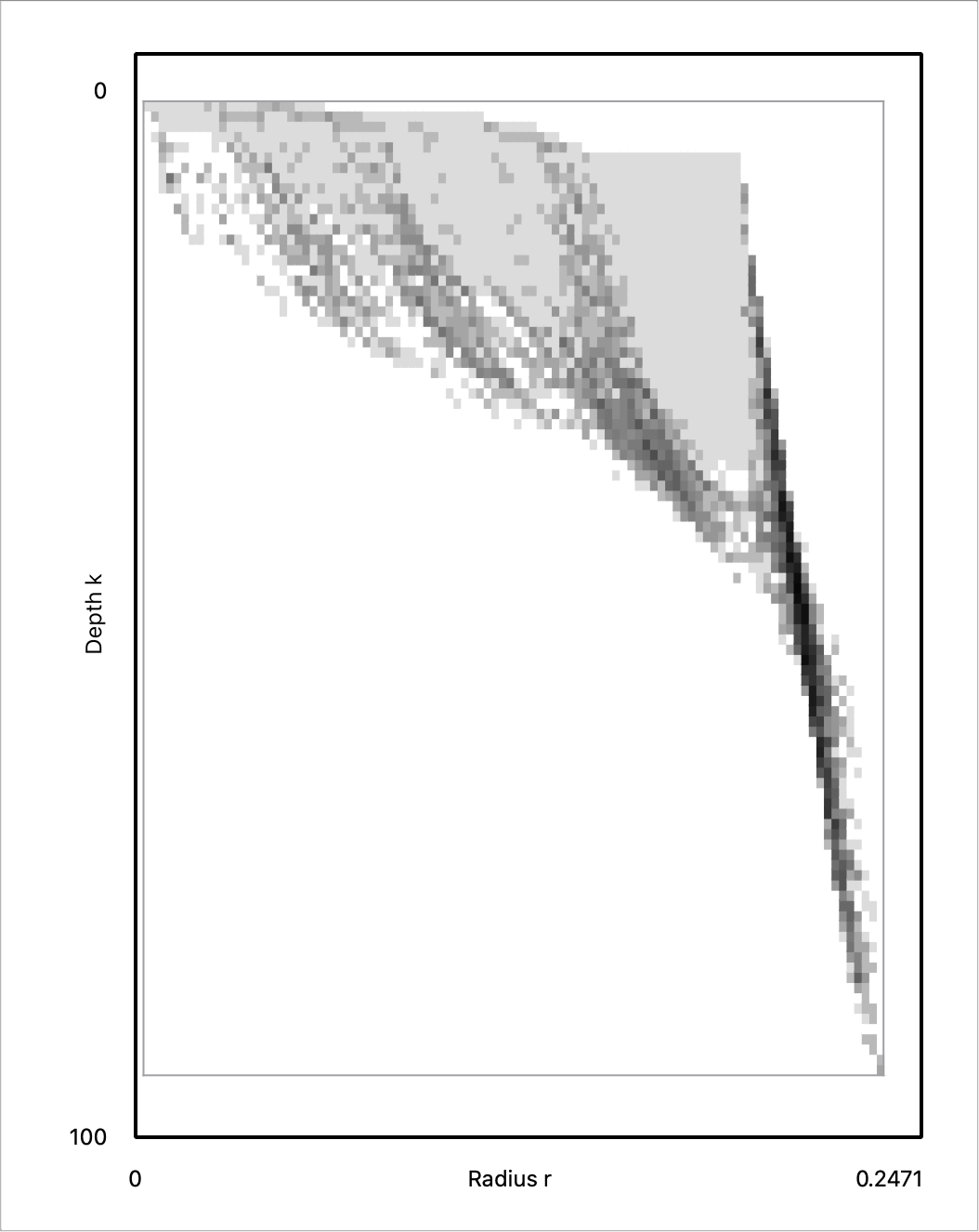} 
	\includegraphics[width=0.305\textwidth]{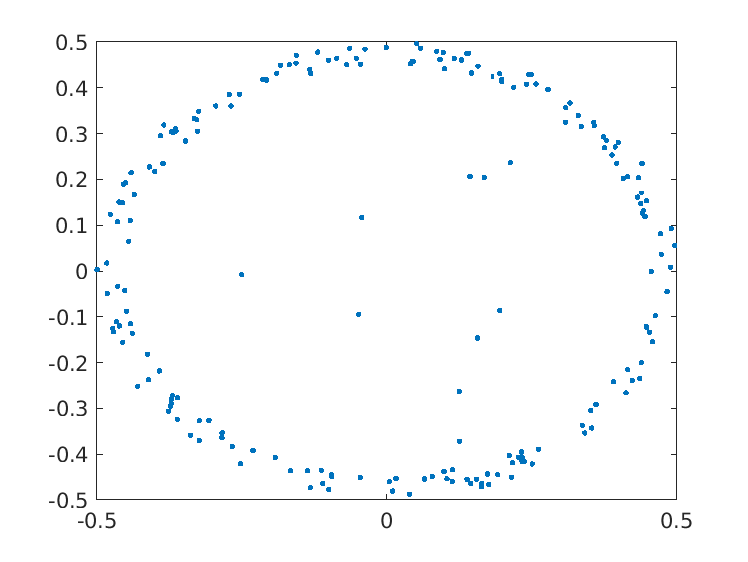} 
	\includegraphics[width=0.18\textwidth]{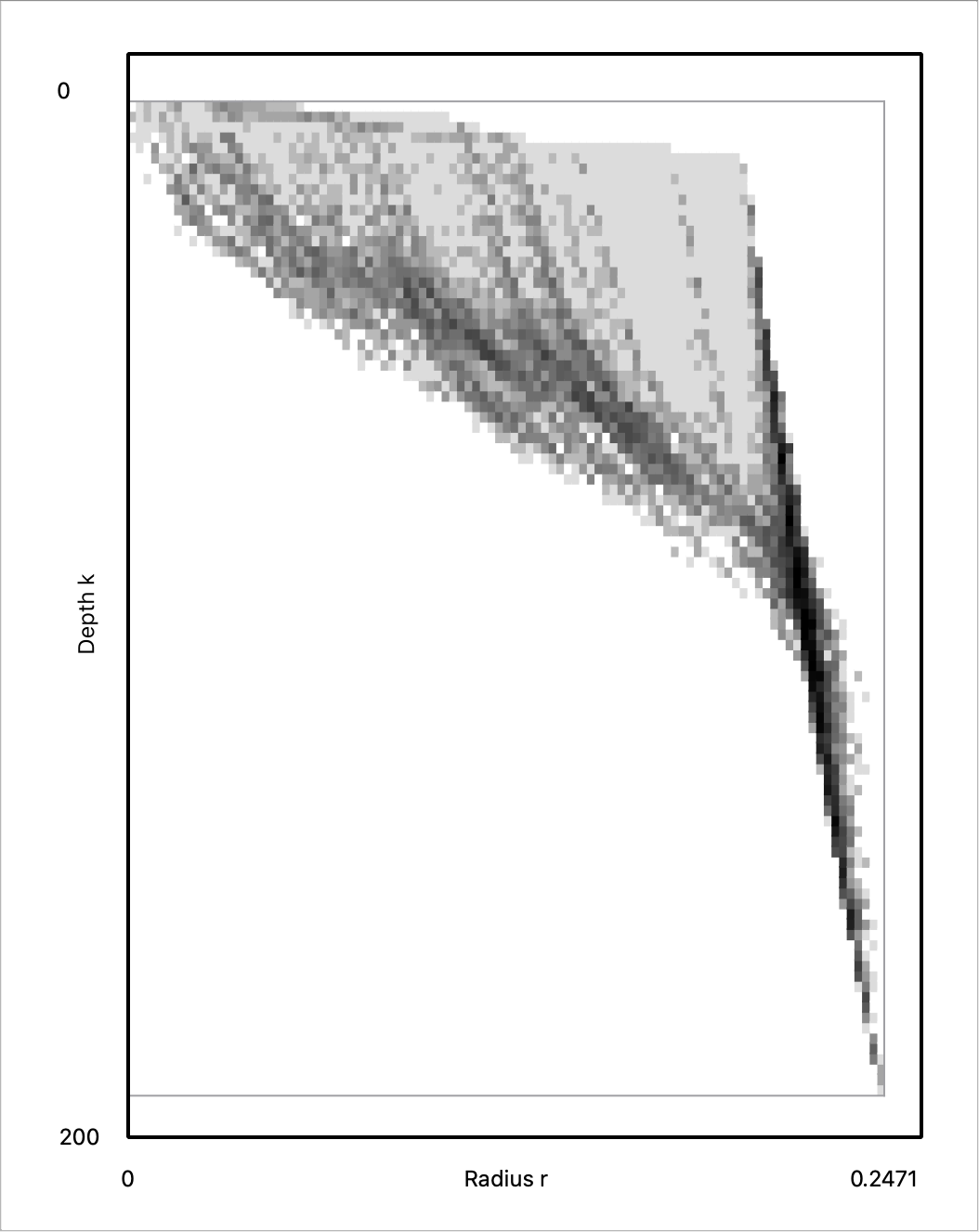} 
	\\ 
	\includegraphics[width=0.305\textwidth]{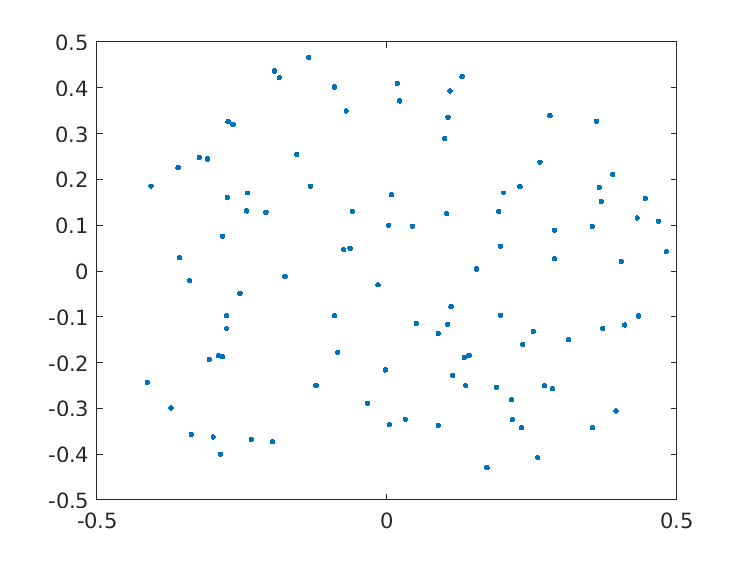} 
	\includegraphics[width=0.18\textwidth]{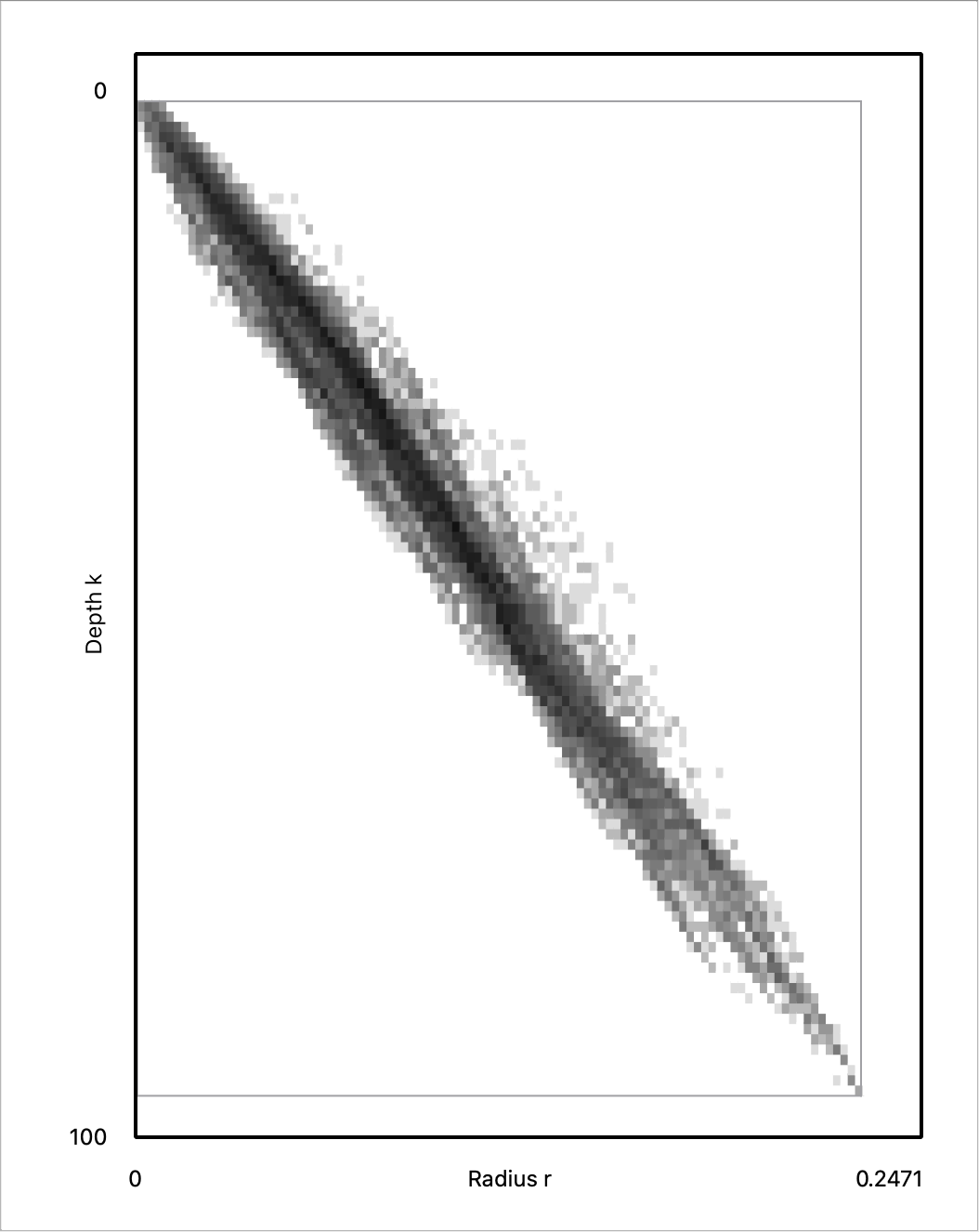} 
	\includegraphics[width=0.305\textwidth]{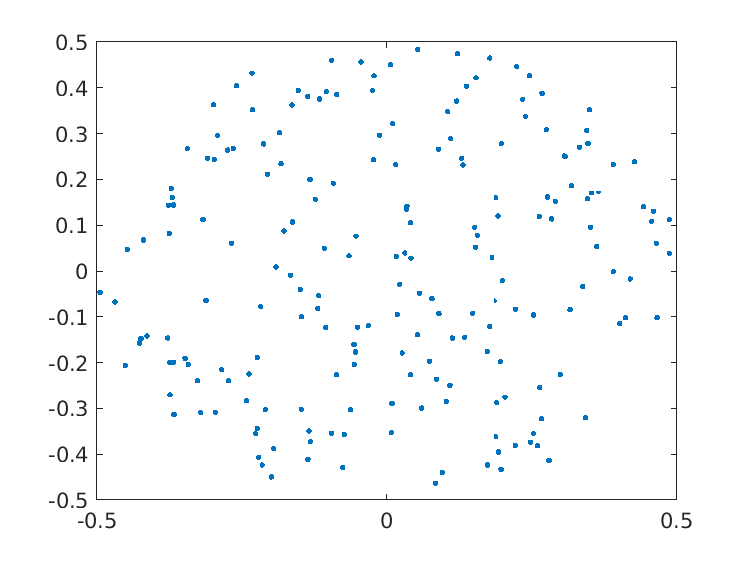} 
	\includegraphics[width=0.18\textwidth]{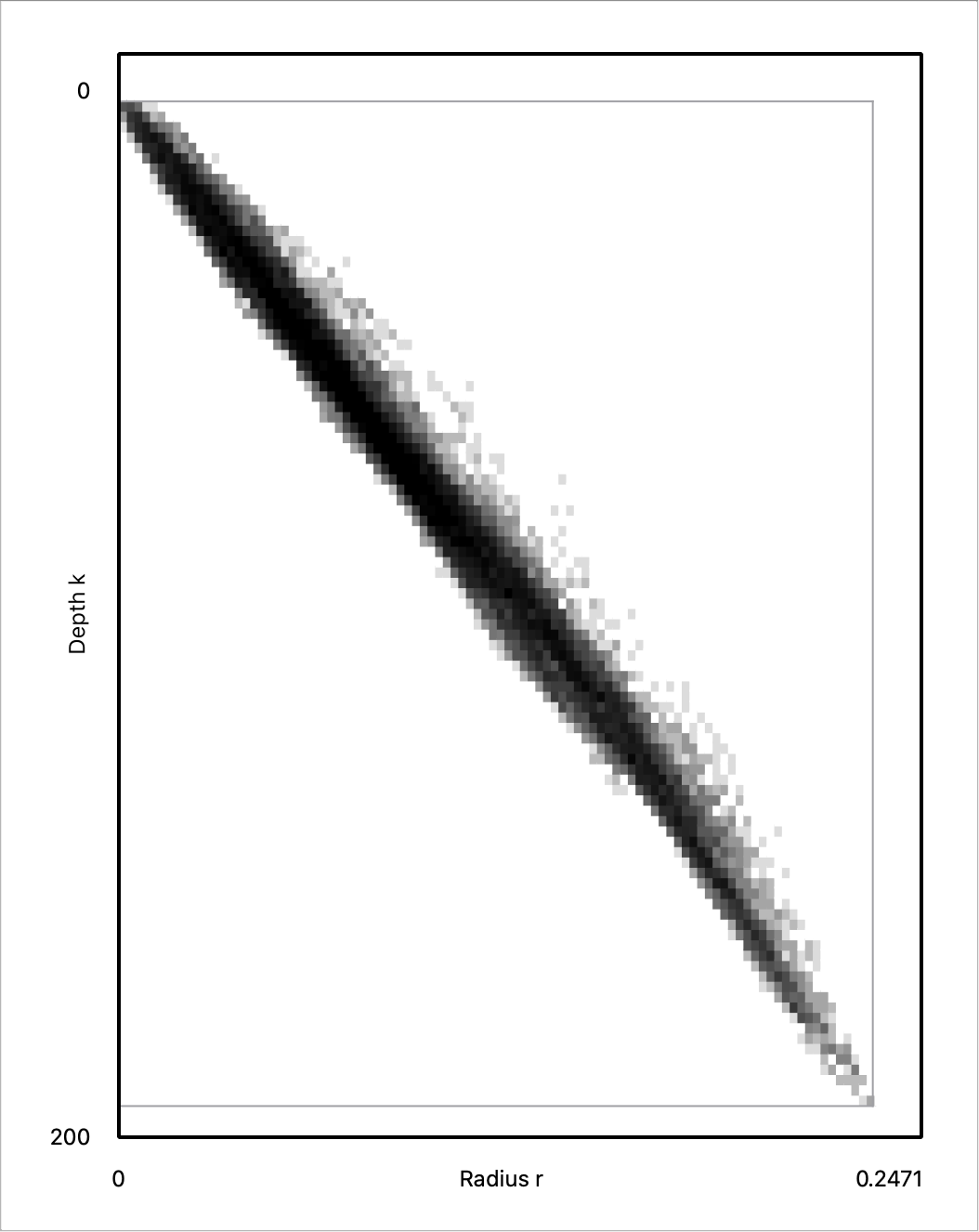} 

	\caption{An illustration of the first Hilbert function of the multicover bifiltration, using grayscale shading.  The instances are samples of an annulus (top), a noisy annulus (middle), and a disk (bottom). The sample size is 100 in the left column, and 200 in the right column.  Darkness of the shading is proportional to the value of the Hilbert function, up to some maximum value, above which the shading is taken to be black; the lightest non-white shade of gray corresponds to a Hilbert function value of 1.} 
	\label{fig:noisy}
\end{figure}

We investigated the size of $\rhomonly^{\leq\kmax}$, 
varying the sample size and the threshold $K$.
For $d=2$, our experiments show a clear subquadratic growth of the size of $\rhomonly^{\leq\kmax}$ and its FIREP with respect to increasing $\kmax$. For $d=3$, the growth is clearly subcubic. These observations also extend to time complexity.
Letting the number of points increase, the size of $\rhomonly^{\leq\kmax}$ and its FIREP shows roughly linear growth for both space dimensions, with a slight superlinear tendency. Again, we observed the same behavior for the computation time.

We conclude this section with a data visualization enabled by the ideas of this paper: For $i\geq 0$, the \emph{$i$-th Hilbert function} assigns to each parameter $(r,k)\in \R\times \N$ the rank of $i$-th homology module of $\cov{r}{k}$ (with coefficients in some fixed field). The Hilbert functions are well known to be unstable invariants.  Nevertheless, their visualization can give us a feel for how the Lipschitz stability property of the multicover bifiltration established in~\cite{blumberg2020stability} manifests itself in random data.  Figure~\ref{fig:noisy} shows a few examples, plotted \linebreak using~\textsc{rivet}\footnote{\url{https://github.com/rivetTDA/}}.


\section{Conclusion}\label{sec:conclusion}
We have introduced a simplicial model for the multicover bifiltration, as well as a polyhedral model based on the rhomboid tiling of \cite{Edelsbrunner2018}.  For a data set of size $n$ in $\R^d$ with $d$ constant, the size of both constructions is $O(n^{d+1})$.  The size can be controlled by thresholding the parameter $k$ of the multicover bifiltration.  An algorithm of \cite{EdOs20} computes the rhomboid bifiltration, and an implementation is available.  In our experimental results, this approach scales well enough to suggest that practical applications could soon be within reach.  A natural next step is to begin exploring the use of the multicover bifiltration on real world data.  

To obtain our combinatorial models of the multicover bifiltration, we begin with a zigzag of filtrations, and then straighten it out by taking unions of prefixes.  Notably, one could in principle compute the persistent homology modules of the multicover bifiltration without straightening out the zigzag, by inverting the isomorphisms on homology induced by the~inclusions $\del{r}{k}\hookrightarrow \tdel{r}{k}$.   It seems plausible that this approach could be computationally~useful.

We are curious to learn which indecomposables typically arise in the \linebreak persistent homology modules of multicover bifiltration, and our approach could be used in conjunction with existing algorithms~\cite{dey2019generalized,holt1998meataxe} to study this.  It would also be interesting to investigate whether there is an interplay between the geometry of a space and the multicover bifiltration of a noisy sample of this space; we wonder if invariants of the bifiltration encode additional information about geometric properties, such as the reach or differentiability. 

Our experiments show a significant increase in the size of our models of multicover bifiltration for increasing $\kmax$.   This suggests the need for refinements to our algorithmic approach in order to handle large values of $\kmax$.  
Aside from the truncations considered in this paper, there are a couple of promising ways forward: One could construct a coarsened bifiltration where some values of $k$ are skipped.  Alternatively, one could make use of the inductive nature of our constructions: for the step from $k$ to $k+1$, one does not need information about the bifiltrations at indices $j<k$. Therefore, one could provide the bifiltration as an output stream without storing it completely in memory. Subsequent algorithmic steps would then have to be implemented as streaming algorithms as well.  

\paragraph{Acknowledgments.} The authors thank the anonymous reviewers for many helpful comments and suggestions, which led to substantial improvements of the paper.  The first two authors were supported by the Austrian Science Fund (FWF) grant number P 29984-N35 and W1230. The first author was partly supported by an Austrian Marshall Plan Scholarship, and by the Brummer \& Partners MathDataLab.
A conference version of this paper was presented at the 37th International Symposium on Computational Geometry (SoCG 2021).

\bibliography{kfoldcov_arxiv.bib}

\begin{thebibliography}{10}

\bibitem{anai2020dtm}
Hirokazu Anai, Fr{\'e}d{\'e}ric Chazal, Marc Glisse, Yuichi Ike, Hiroya
  Inakoshi, Rapha{\"e}l Tinarrage, and Yuhei Umeda.
\newblock {D}{T}{M}-based filtrations.
\newblock In {\em 35th International Symposium on Computational Geometry (SoCG
  2019)}. Schloss Dagstuhl--Leibniz-Zentrum fuer Informatik, 2019.
\newblock \href {https://doi.org/10.4230/LIPIcs.SoCG.2019.58}
  {\path{doi:10.4230/LIPIcs.SoCG.2019.58}}.

\bibitem{aurenhammer1987power}
Franz Aurenhammer.
\newblock Power diagrams: properties, algorithms and applications.
\newblock {\em SIAM Journal on Computing}, 16(1):78--96, 1987.
\newblock \href {https://doi.org/10.1137/0216006} {\path{doi:10.1137/0216006}}.

\bibitem{kerber2020}
Ulrich Bauer, Michael Kerber, Fabian Roll, and Alexander Rolle.
\newblock A unified view on the functorial nerve theorem and its variations.
\newblock 2022.
\newblock \href {http://arxiv.org/abs/2203.03571} {\path{arXiv:2203.03571}}.

\bibitem{blumberg2014robust}
Andrew~J. Blumberg, Itamar Gal, Michael~A. Mandell, and Matthew Pancia.
\newblock Robust statistics, hypothesis testing, and confidence intervals for
  persistent homology on metric measure spaces.
\newblock {\em Foundations of Computational Mathematics}, 14(4):745--789, 2014.
\newblock \href {https://doi.org/10.1007/s10208-014-9201-4}
  {\path{doi:10.1007/s10208-014-9201-4}}.

\bibitem{blumberg2017universality}
Andrew~J. Blumberg and Michael Lesnick.
\newblock Universality of the homotopy interleaving distance, 2017.
\newblock \href {http://arxiv.org/abs/1705.01690} {\path{arXiv:1705.01690}}.

\bibitem{blumberg2020stability}
Andrew~J. Blumberg and Michael Lesnick.
\newblock Stability of 2-parameter persistent homology, 2020.
\newblock \href {http://arxiv.org/abs/2010.09628} {\path{arXiv:2010.09628}}.

\bibitem{bobrowski2017topological}
Omer Bobrowski, Sayan Mukherjee, Jonathan~E. Taylor, et~al.
\newblock Topological consistency via kernel estimation.
\newblock {\em Bernoulli}, 23(1):288--328, 2017.
\newblock \href {https://doi.org/10.3150/15-BEJ744}
  {\path{doi:10.3150/15-BEJ744}}.

\bibitem{boissonnat2018geometric}
Jean-Daniel Boissonnat, Fr{\'e}d{\'e}ric Chazal, and Mariette Yvinec.
\newblock {\em Geometric and topological inference}, volume~57.
\newblock Cambridge University Press, 2018.

\bibitem{Buchet15}
Micka{\"{e}}l Buchet, Fr{\'{e}}d{\'{e}}ric Chazal, Steve~Y. Oudot, and
  Donald~R. Sheehy.
\newblock Efficient and robust persistent homology for measures.
\newblock In {\em Proceedings of the 26th Annual {ACM-SIAM} Symposium on
  Discrete Algorithms, (SODA 2015)}, 2015.
\newblock \href {https://doi.org/10.1137/1.9781611973730.13}
  {\path{doi:10.1137/1.9781611973730.13}}.

\bibitem{carlsson2008local}
Gunnar Carlsson, Tigran Ishkhanov, Vin De~Silva, and Afra Zomorodian.
\newblock On the local behavior of spaces of natural images.
\newblock {\em International journal of computer vision}, 76(1):1--12, 2008.
\newblock \href {https://doi.org/10.1007/s11263-007-0056-x}
  {\path{doi:10.1007/s11263-007-0056-x}}.

\bibitem{carlsson2009theory}
Gunnar Carlsson and Afra Zomorodian.
\newblock The theory of multidimensional persistence.
\newblock {\em Discrete \& Computational Geometry}, 42(1):71--93, 2009.
\newblock \href {https://doi.org/10.1007/s00454-009-9176-0}
  {\path{doi:10.1007/s00454-009-9176-0}}.

\bibitem{cavanna17when}
Nicholas~J. Cavanna, Kirk~P. Gardner, and Donald~R. Sheehy.
\newblock When and why the topological coverage criterion works.
\newblock In {\em Proceedings of the 28th Annual {ACM-SIAM} Symposium on
  Discrete Algorithms, (SODA 2017)}, 2017.
\newblock \href {https://doi.org/10.1137/1.9781611974782.177}
  {\path{doi:10.1137/1.9781611974782.177}}.

\bibitem{cavanna15geometric}
Nicholas~J. Cavanna, Mahmoodreza Jahanseir, and Donald~R. Sheehy.
\newblock A geometric perspective on sparse filtrations.
\newblock In {\em Proceedings of the 27th Canadian Conference on Computational
  Geometry ({CCCG} 2015)}, 2015.

\bibitem{cerri2013betti}
Andrea Cerri, Barbara~Di Fabio, Massimo Ferri, Patrizio Frosini, and Claudia
  Landi.
\newblock Betti numbers in multidimensional persistent homology are stable
  functions.
\newblock {\em Mathematical Methods in the Applied Sciences},
  36(12):1543--1557, 2013.
\newblock \href {https://doi.org/10.1002/mma.2704}
  {\path{doi:10.1002/mma.2704}}.

\bibitem{chan1997primal}
Timothy~M Chan, Jack Snoeyink, and Chee-Keng Yap.
\newblock Primal dividing and dual pruning: Output-sensitive construction of
  four-dimensional polytopes and three-dimensional voronoi diagrams.
\newblock {\em Discrete \& Computational Geometry}, 18(4):433--454, 1997.

\bibitem{chazal2009analysis}
Fr\'{e}d\'{e}ric Chazal, Leonidas~J. Guibas, Steve~Y. Oudot, and Primoz Skraba.
\newblock Scalar field analysis over point cloud data.
\newblock {\em Discrete \& Computational Geometry}, 46(4):743--775, 2011.
\newblock \href {https://doi.org/10.1007/s00454-011-9360-x}
  {\path{doi:10.1007/s00454-011-9360-x}}.

\bibitem{chazal2013clustering}
Fr{\'{e}}d{\'{e}}ric Chazal, Leonidas~J. Guibas, Steve~Y. Oudot, and Primoz
  Skraba.
\newblock Persistence-based clustering in {Riemannian} manifolds.
\newblock {\em Journal of the ACM}, 60(6), 2013.
\newblock \href {https://doi.org/10.1145/2535927} {\path{doi:10.1145/2535927}}.

\bibitem{chazal2008towards}
Fr\'{e}d\'{e}ric Chazal and Steve~Y. Oudot.
\newblock Towards persistence-based reconstruction in {E}uclidean spaces.
\newblock In {\em 24th International Symposium on Computational Geometry (SoCG
  2008)}, page 232–241. Association for Computing Machinery (ACM), 2008.
\newblock \href {https://doi.org/10.1145/1377676.1377719}
  {\path{doi:10.1145/1377676.1377719}}.

\bibitem{Chazal2011}
Frédéric Chazal, David Cohen-Steiner, and Quentin Mérigot.
\newblock Geometric inference for probability measures.
\newblock {\em Foundations of Computational Mathematics}, 11:733--751, 12 2011.
\newblock \href {https://doi.org/10.1007/s10208-011-9098-0}
  {\path{doi:10.1007/s10208-011-9098-0}}.

\bibitem{clarkson1989applications}
Kenneth~L. Clarkson and Peter~W. Shor.
\newblock Applications of random sampling in computational geometry, {I}{I}.
\newblock {\em Discrete \& Computational Geometry}, 4(5):387--421, 1989.
\newblock \href {https://doi.org/10.1007/BF02187740}
  {\path{doi:10.1007/BF02187740}}.

\bibitem{Cohen-SteinerEdelsbrunnerHarer2007}
David Cohen-Steiner, Herbert Edelsbrunner, and John Harer.
\newblock Stability of persistence diagrams.
\newblock {\em Discrete \& Computational Geometry}, 37(1):103--120, 2007.
\newblock \href {https://doi.org/10.1007/s00454-006-1276-5}
  {\path{doi:10.1007/s00454-006-1276-5}}.

\bibitem{dey2019generalized}
Tamal~K. Dey and Cheng Xin.
\newblock Generalized persistence algorithm for decomposing multi-parameter
  persistence modules, 2020.
\newblock \href {http://arxiv.org/abs/1904.03766} {\path{arXiv:1904.03766}}.

\bibitem{dwyer1995homotopy}
William~G. Dwyer and Jan Spalinski.
\newblock Homotopy theories and model categories.
\newblock {\em Handbook of algebraic topology}, 73:126, 1995.
\newblock \href {https://doi.org/10.1016/B978-044481779-2/50003-1}
  {\path{doi:10.1016/B978-044481779-2/50003-1}}.

\bibitem{edelsbrunner1987algorithms}
Herbert Edelsbrunner.
\newblock {\em Algorithms in combinatorial geometry}.
\newblock Springer-Verlag, 1987.
\newblock \href {https://doi.org/10.1007/978-3-642-61568-9}
  {\path{doi:10.1007/978-3-642-61568-9}}.

\bibitem{Edelsbrunner1995}
Herbert Edelsbrunner.
\newblock The union of balls and its dual shape.
\newblock {\em Discrete {\&} Computational Geometry}, 13(3):415--440, 1995.
\newblock \href {https://doi.org/10.1007/BF02574053}
  {\path{doi:10.1007/BF02574053}}.

\bibitem{Edelsbrunner_shape_2006}
Herbert Edelsbrunner.
\newblock Shape reconstruction with {D}elaunay complex.
\newblock In {\em LATIN'98: Theoretical Informatics}, pages 119--132. Springer
  Berlin Heidelberg, 1998.
\newblock \href {https://doi.org/10.1007/BFb0054315}
  {\path{doi:10.1007/BFb0054315}}.

\bibitem{edelsbrunner2010computational}
Herbert Edelsbrunner and John Harer.
\newblock {\em Computational topology: an introduction}.
\newblock American Mathematical Society, 2010.
\newblock \href {https://doi.org/10.1007/978-3-540-33259-6_7}
  {\path{doi:10.1007/978-3-540-33259-6_7}}.

\bibitem{Edelsbrunner2018}
Herbert Edelsbrunner and Georg Osang.
\newblock The multi-cover persistence of {E}uclidean balls.
\newblock In {\em 34th International Symposium on Computational Geometry (SoCG
  2018)}. Schloss Dagstuhl-Leibniz-Zentrum fuer Informatik, 2018.
\newblock \href {https://doi.org/10.4230/LIPIcs.SoCG.2018.34}
  {\path{doi:10.4230/LIPIcs.SoCG.2018.34}}.

\bibitem{EdOs20}
Herbert Edelsbrunner and Georg Osang.
\newblock A simple algorithm for computing higher order {D}elaunay mosaics and
  $\alpha$-shapes, 2020.
\newblock \href {http://arxiv.org/abs/2011.03617} {\path{arXiv:2011.03617}}.

\bibitem{Edelsbrunner1986}
Herbert Edelsbrunner and Raimund Seidel.
\newblock Voronoi diagrams and arrangements.
\newblock {\em Discrete {\&} Computational Geometry}, 1(1):25--44, 1986.
\newblock \href {https://doi.org/10.1007/BF02187681}
  {\path{doi:10.1007/BF02187681}}.

\bibitem{guibas2013witnessed}
Leonidas Guibas, Dmitriy Morozov, and Quentin M{\'e}rigot.
\newblock Witnessed k-distance.
\newblock {\em Discrete \& Computational Geometry}, 49(1):22--45, 2013.
\newblock \href {https://doi.org/10.1007/s00454-012-9465-x}
  {\path{doi:10.1007/s00454-012-9465-x}}.

\bibitem{harrington2017stratifying}
Heather~A. Harrington, Nina Otter, Hal Schenck, and Ulrike Tillmann.
\newblock Stratifying multiparameter persistent homology, 2017.
\newblock \href {http://arxiv.org/abs/1708.07390} {\path{arXiv:1708.07390}}.

\bibitem{hatcher2005algebraic}
Allen Hatcher.
\newblock {\em Algebraic Topology}.
\newblock Cambridge University Press, 2005.

\bibitem{hirschhorn2009model}
Philip~S. Hirschhorn.
\newblock {\em Model categories and their localizations}, volume~99.
\newblock American Mathematical Society, 2009.

\bibitem{holt1998meataxe}
Derek~F. Holt.
\newblock The {M}eataxe as a tool in computational group theory.
\newblock {\em London Mathematical Society Lecture Note Series}, pages 74--81,
  1998.

\bibitem{kerber2021fast}
Michael Kerber and Alexander Rolle.
\newblock Fast minimal presentations of bi-graded persistence modules.
\newblock In {\em 2021 Proceedings of the Symposium on Algorithm Engineering
  and Experiments (ALENEX)}, pages 207--220. SIAM, 2021.
\newblock \href {https://doi.org/10.1137/1.9781611976472.16}
  {\path{doi:10.1137/1.9781611976472.16}}.

\bibitem{kozlov-book}
Dmitry Kozlov.
\newblock {\em Combinatorial Algebraic Topology}.
\newblock Springer, 2008.
\newblock \href {https://doi.org/10.1007/978-3-540-71962-5}
  {\path{doi:10.1007/978-3-540-71962-5}}.

\bibitem{Krasnoshchekov}
Dmitry Krasnoshchekov and Valentin Polishchuk.
\newblock Order-k alpha-hulls and alpha-shapes.
\newblock {\em Information Processing Letters}, 114(1-2):76--83, 2014.
\newblock \href {https://doi.org/10.1016/j.ipl.2013.07.023}
  {\path{doi:10.1016/j.ipl.2013.07.023}}.

\bibitem{lanari2020rectification}
Edoardo Lanari and Luis~N. Scoccola.
\newblock Rectification of interleavings and a persistent {W}hitehead theorem,
  2020.
\newblock \href {http://arxiv.org/abs/2010.05378} {\path{arXiv:2010.05378}}.

\bibitem{leray45}
Jean Leray.
\newblock Sur la forme des espaces topologiques et sur les points fixes des
  representations.
\newblock {\em Journal de Math{\'e}matiques Pures et Appliqu{\'e}es}, 24, 01
  1945.

\bibitem{Lesnick2015}
Michael Lesnick and Matthew Wright.
\newblock Interactive visualization of {2-D} persistence modules, 2015.
\newblock \href {http://arxiv.org/abs/1512.00180} {\path{arXiv:1512.00180}}.

\bibitem{lesnick2019computing}
Michael Lesnick and Matthew Wright.
\newblock Computing minimal presentations and {B}etti numbers of 2-parameter
  persistent homology, 2019.
\newblock \href {http://arxiv.org/abs/1902.05708} {\path{arXiv:1902.05708}}.

\bibitem{osang2021multi}
Georg~F. Osang.
\newblock {\em {Multi-cover persistence and {D}elaunay mosaics}}.
\newblock PhD thesis, IST Austria, 2021.
\newblock \href {https://doi.org/10.15479/AT:ISTA:9056}
  {\path{doi:10.15479/AT:ISTA:9056}}.

\bibitem{phillips2015geometric}
Jeff~M. Phillips, Bei Wang, and Yan Zheng.
\newblock Geometric inference on kernel density estimates.
\newblock In {\em 31st International Symposium on Computational Geometry (SoCG
  2015)}. Schloss Dagstuhl-Leibniz-Zentrum fuer Informatik, 2015.
\newblock \href {https://doi.org/10.4230/LIPIcs.SOCG.2015.857}
  {\path{doi:10.4230/LIPIcs.SOCG.2015.857}}.

\bibitem{scoccola2020locally}
Luis~N. Scoccola.
\newblock {\em Locally Persistent Categories And Metric Properties Of
  Interleaving Distances}.
\newblock PhD thesis, The University of Western Ontario, 2020.
\newblock URL: \url{https://ir.lib.uwo.ca/etd/7119/}.

\bibitem{Sheehy2012}
Donald~R. Sheehy.
\newblock A multicover nerve for geometric inference.
\newblock In {\em Proceedings of the 24th Canadian Conference in Computational
  Geometry ({CCCG} 2012)}, 2012.

\bibitem{Sheehy2013}
Donald~R. Sheehy.
\newblock Linear-size approximations to the {V}ietoris--{R}ips filtration.
\newblock {\em Discrete {\&} Computational Geometry}, 49(4):778--796, 2013.
\newblock \href {https://doi.org/10.1007/s00454-013-9513-1}
  {\path{doi:10.1007/s00454-013-9513-1}}.

\bibitem{sheehy21sparse}
Donald~R. Sheehy.
\newblock A sparse delaunay filtration, 2020.
\newblock \href {http://arxiv.org/abs/2012.01947} {\path{arXiv:2012.01947}}.

\bibitem{vipond2018multiparameter}
Oliver Vipond.
\newblock Multiparameter persistence landscapes.
\newblock {\em Journal of Machine Learning Research}, 21(61):1--38, 2020.
\newblock URL: \url{http://jmlr.org/papers/v21/19-054.html}.

\bibitem{voronoi1908recherches}
Georgy Voronoi.
\newblock Recherches sur les parall{\'e}lo{\`e}dres primitives.
\newblock {\em Journal f{\"u}r die reine und angewandte Mathematik},
  134:198--287, 1908.

\bibitem{ZomorodianCarlsson2005}
Afra Zomorodian and Gunnar Carlsson.
\newblock Computing persistent homology.
\newblock {\em Discrete and Computational Geometry}, 33(2):249--274, 2005.
\newblock \href {https://doi.org/10.1007/s00454-004-1146-y}
  {\path{doi:10.1007/s00454-004-1146-y}}.

\end{thebibliography}


\appendix

\section{Details on experiments}
\label{app:experiments}

\subsection{Implementation}

Let us mention an important technicality in our pipeline for computing minimal presentations: In order to limit the size of the minimal presentations, we ``snap'' the radius values of all generators and relations onto a set of 100 evenly spaced points in $\R$.  The values of the parameter $k$ are left unchanged.  This snapping is done only after the minimal presentation is computed.  The snapping process in fact can make a minimal presentation non-minimal, so after snapping, we re-minimize the presentation.  All reported results below are for minimal presentations computed using this pipeline.  The Hilbert functions shown in Figure~\ref{fig:noisy} were also computed from such ``snapped" presentations.

\subsection{Experimental output} 

We present the concrete outcome of some of our experiments.
All results are averaged
over $5$ runs with independently generated data sets.  The sizes reflect  the number of elements of the corresponding set, and the times were measured in seconds. 

We give a brief overview of the experiments, referring to the tables for further details.  We were curious about the practical improvements from $\trhomonly$ to $\rhomonly$. We documented these for a few values of $n$ and $\kmax$ in the plane. The sizes of their truncated versions both grow linearly in the number of points. See Table~\ref{tab:srhombrhomb} for more refined results. In further experiments, we only used $\rhomonly$.

Table~\ref{tab:sizesKfix} shows the behavior of fixed $\kmax$ and an increasing number of uniformly sampled points. Conversely, Table~\ref{tab:sizesNfix} documents the behavior of $\rhomonly$ in an experiment with a fixed number of points and increasing $\kmax$. Both in dimension $2$ and $3$, we investigated the size of the bifiltration, the size of the FIREPs, and the size of minimal presentations thereof. We also kept track of the time needed for the computations.

Finally, we wondered how the measurements change when data sets are sampled from a particular shape. As an example, we sampled points from an annulus with random but bounded perturbations and added uniform background noise to it. Letting both the range of the perturbations and the portion of the background noise vary, we tracked the size of the minimal presentations in Table~\ref{tab:noisycircle}.

\begin{table}[hbt!]
\centering
\scalebox{.85}{
\begin{tabular}{|c|c|c|c|c|c|c|c|c|c|c|c|}
\hline
 \multicolumn{2}{|c|}{$d=2$} & \multicolumn{3}{c}{FIREP ($\trhomonly^{\leq\kmax}$)} & \multicolumn{3}{|c|}{FIREP ($\rhomonly^{\leq\kmax}$)} &  \multicolumn{2}{c|}{} \\
\hline
$n$ & $\kmax$ 
& size & time & $\frac{\textnormal{size}}{n\cdot\kmax^2}$
& size  &  time & $\frac{\textnormal{size}}{n\cdot\kmax^2}$ & 
$\frac{\textnormal{size (u)}}{\textnormal{size (s)}}$ & $\frac{\textnormal{time (u)}}{\textnormal{time (s)}}$ \\
\hline
10,000 & 4 
&1.916 M & 25.67 & 12.0
& 1.197 M & 19.49 &  7.48
 & 62\% &  76\% \\
20,000 & 4 
&3.836 M & 50.79 & 12.0
& 2.397 M & 40.11 & 7.49
 & 62\% & 79\% \\
40,000 & 4 
& 7.676 M & 107.44 & 12.0
& 4.797 M & 84.21 & 7.49
 & 62\% & 78\% \\
80,000 & 4 
&15.355 M & 228.29 & 12.0
& 8.597 M & 179.45 & 6.71
 & 56\% & 78\% \\
10,000 & 8 
& 8.249 M & 106.90 & 12.9
& 4.344 M & 77.30&  6.79
 & 53\% & 72\% \\
20,000 & 8
& 16.528 M & 231.75 & 12.9
& 8.702 M & 166.55& 6.80
 & 53\% & 72\% \\
40,000 & 8 
&17.709 M & 475.51 & 6.95
& 15.265 M & 346.67 & 5.97
 & 86\% & 73 \% \\
80,000 & 8 
&
17.718 M & 1,010.44 & 3.46
& 13.835 M & 746.48& 2.70
 & 78\% & 74\% \\
\hline
\end{tabular}}
\caption{We compare the sizes of FIREPs of $\trhomonly^{\leq\kmax}$ and $\rhomonly^{\leq\kmax}$. We considered $n$ uniformly sampled points in the unit square. The FIREPs of the unsliced bifiltration show 63\% of the size of FIREPs of the sliced version for $d=2$, and 54\% for  $d=3$. That advantage seems to get smaller for a sufficiently large number of points.  Computing the unsliced version has revealed to be more than 20\% faster for $K=4$ and to be more than 25\% faster for $K=8$ in this experiment. A generalization of these findings would describe the computational value of $\rhomonly$ more accurately. 
}
\label{tab:srhombrhomb}
\end{table}

\begin{table}[hbt!]
\centering
\scalebox{.85}{
\begin{tabular}{|c|c|c|c|c|c|c|c|c|c|c|}
\hline
 \multicolumn{2}{|c|}{} & \multicolumn{2}{|c|}{$\rhomonly^{\leq\kmax}$} & \multicolumn{4}{c|}{FIREP} & \multicolumn{2}{c|}{snapped minpres} \\
\hline
$d$ & $n$ & size &  $\frac{\textnormal{size}}{n\cdot\kmax^3}$  &  size &  $\frac{\textnormal{size}}{n\cdot\kmax^2}$  &  time & $\frac{10^4\cdot\textnormal{time}}{ n\cdot\kmax^3}$ & size & time \\
\hline
2 & 1,250 & 0.192 M & 9.60 & 0.148 M & 7.40 & 1.95 & 3.90 & 0.0182 M & 1.58  \\
2 & 2,500 & 0.387 M & 9.68 & 0.297 M & 7.43 & 6.52 & 1.60 & 0.0373 M & 4.56 \\
2 & 5,000 & 0.777 M & 9.71 & 0.598 M & 7.48 & 9.30 & 1.16 & 0.0749 M & 6.61 \\
2 & 10,000 & 1.56 M & 9.75 & 12.0 M & 7.50 & 20.13 & 1.26 & 0.152 M & 15.36 \\
2 & 20,000 & 3.12 M & 9.75 & 2.40 M & 7.5 & 40.36 & 1.26 & 0.281 M & 32.39  \\
\hline
$d$ & $n$ & size &  $\frac{\textnormal{size}}{n\cdot\kmax^3}$  &  size &  $\frac{\textnormal{size}}{n\cdot\kmax^3}$  &  time & $\frac{10^4\cdot\textnormal{time}}{5\cdot n\cdot\kmax^3}$ & size & time \\
\hline
3 & 1,250 & 2.28 M &  28.50 & 1.42 M & 17.75 & 45.39 & 1.13 & 0.0806 M & 10.39 \\
3 & 2,500 & 4.68 M &  29.25 & 2.92 M & 18.25 & 95.76 & 1.20 & 0.176 M & 24.34 \\
3 & 5,000 & 9.49 M &  29.66 & 5.91 M & 18.47 & 202.92 & 1.27 & 0.360 M & 58.35 \\
3 & 10,000 & 19.2 M & 30.00 & 12.0 M & 18.75 & 431.58 & 1.35 & 0.739 M & 142.12  \\
3 & 20,000 & 38.7 M & 30.23 & 24.1 M & 18.83 & 904.42 & 1.41 & 1.52 M & 378.25 \\
\hline
\end{tabular}}
\caption{Results on $n$ uniformly sampled points in the unit square and unit cube, using $\kmax=4$ as maximal value for $k$. The size of the bifiltration and the FIREP grows slightly superlinear in the documented scale window. The computation time seems to be of the same complexity.}
\label{tab:sizesKfix}
\end{table}

\begin{table}[hbt!]
\centering
\scalebox{.82}{
\begin{tabular}{|c|c|c|c|c|c|c|c|c|c|c|c|c|}
\hline
 \multicolumn{2}{|c|}{} & \multicolumn{2}{c|}{$\rhomonly^{\leq\kmax}$} & \multicolumn{4}{c}{FIREP} & \multicolumn{3}{|c|}{snapped minpres} \\
\hline
$d$ & $\kmax$ & size & $\frac{\textnormal{size}}{n\cdot\kmax^2}$ &  size &  $\frac{\textnormal{size}}{n\cdot\kmax^2}$ &  time & $\frac{10^5\cdot\textnormal{time}}{n\cdot\kmax^2}$ & size & time & $\frac{\textnormal{size}}{\textnormal{FIREP size}}$ \\
\hline
2 & 2 & 0.0216  M & 10.8 & 0.0167 M & 8.35 & 0.26 &  1.30  & 1.67 K & 0.29 & 10.0\% \\
2 & 4 & 0.0758 M  & 9.48 & 0.0583 M & 7.29 & 0.76 &  0.95  & 7.08 K & 0.74 &  12.1\% \\
2 & 8 &  0.272 M & 8.50 & 0.207 M & 6.47 & 3.14 &  0.98  & 23.6 K & 2.13 &  11.4\%\\
2 & 16 & 0.997 M  & 7.79 & 0.755 M & 5.90 &  12.53  & 0.98  & 74.7 K & 7.47 &  9.89\% \\
\hline
$d$ & $\kmax$ & size &  $\frac{\textnormal{size}}{n\cdot\kmax^3}$  &  size &  $\frac{\textnormal{size}}{n\cdot\kmax^3}$  &  time & $\frac{10^5\cdot\textnormal{time}}{5\cdot n\cdot\kmax^3}$ & size & time & $\frac{\textnormal{size}}{\textnormal{FIREP size}}$ \\
\hline
3 & 2 &  0.144 M  & 36.0 &  0.0860 M  & 21.5 &  2.30 &  1.15  & 5.70 K & 0.86 &  6.63\% \\
3 & 4 &  0.863 M &  27.0 & 0.539 M& 16.8 & 16.1 & 1.00 & 30.8 K & 3.74 & 5.71\% \\
3 & 8 &  5.33 M & 20.8 & 3.36 M  & 13.1 &   122  & 0.95 & 146 K & 23.7 & 4.34\% \\
3 & 16 &  34.0 M & 16.6 & 21.4 M & 10.4 &  1,006  & 0.98 & 627 K & 192 & 2.93\% \\
\hline
\end{tabular}}
\caption{Results on 500 uniformly sampled points in the unit square and unit cube. The sizes of the bifiltration and its FIREP grow clearly subquadratic for $d=2$, and clearly subcubic for $d=3$. The running times seem close to being quadratic and cubic in dimensions 2 and 3, respectively. We also measured the relation of the sizes of the snapped minimal presentations and FIREPs.}
\label{tab:sizesNfix}
\end{table}

\begin{table}[hbt!]
\centering
\scalebox{.88}{
\begin{tabular}{|c|c|c|c|}
\hline
p & $err$ & snapped minpres size & relative size \\
\hline
1 & 0.01 &  2,015  
& 0.38 \% \\
1 & 0.04 &  1,691 
&  0.32 \% \\
1 & 0.08 &  1,526
& 0.29 \% \\
1 & 0.12 &   5,497 
& 1.0 \% \\
1 & 0.14 &    16,183 
& 3.1 \% \\
1 & 0.16 &  29,988 
& 5.7 \% \\
4 & 0.01 &  14,389  
& 2.7 \% \\
4 & 0.04 &   12,852 
& 2.4 \% \\
4 & 0.08 &  15,198 
& 2.9 \% \\
4 & 0.12 &  18,457  
& 3.5 \% \\
4 & 0.16 &  43,917  
& 8.3 \% \\
16 & 0.01 &  70,901 
& 13 \% \\
16 & 0.04 &   84,142 
& 16 \% \\
16 & 0.08 &  99,782 
& 19\% \\
16 & 0.16 &  146,170 
& 28\% \\
64 & 0.01 &   344,847 
& 65 \% \\
64 & 0.04 &  365,308 
& 69 \% \\
64 & 0.08 &  416,522 
& 79 \% \\
64 & 0.16 &   427,697 
& 80 \% \\
\hline
100 & - &  529,128  
& 100 \% \\
\hline
\end{tabular}}
\caption{We sampled points of an annulus around a circle of radius $0.25$. In total, we have 10,000 points whereas $p\%$ of these points are uniform noise in the surrounding box $[0,1]^2$. The other points are sampled with a random perturbation per coordinate bounded by a number $err$. We considered $\kmax=8$ as maximal value for $k$. The size of the snapped minimal presentation increases when adding more uniform noise. It may increase more drastically within a certain range of $p$, i.e., for $p\in\{1,4\}$, starting at about at 0.12. We also observed a considerable variance of the individual results in such areas. In particular, the size of the snapped minimal presentations is neither a linear, nor a sub- or superlinear process. We regard this process mostly as a property of the snapping technique. Furthermore, when $p$ is not too big, the perturbations around each sampled point can be quite high, i.e., for $p\leq 4$ and $err=0.16$, the snapped minimal presentations are still of relative size $5.7\%$ and $8.3\%$, respectively. Note that the samples only stay inside the surrounding box  $[0,1]^2$ if $err\leq 0.25$.
}
\label{tab:noisycircle}
\end{table}

\end{document}